\documentclass[10pt]{article}
\usepackage{import} 
\usepackage[T1]{fontenc} 
\usepackage[margin=2.1cm]{geometry}
\usepackage[l2tabu,orthodox]{nag}
\usepackage[all,warning]{onlyamsmath}

\usepackage[english]{babel}
\usepackage[babel=true,final]{microtype}
\usepackage{mathtools}
\usepackage{amsfonts}
\usepackage{amssymb}
\usepackage{amsthm}
\usepackage{bm}
\usepackage{adjustbox}
\usepackage{svg}

\usepackage[
    backend=biber,
    citestyle=authoryear,
    style=alphabetic,
    giveninits=true,
    url=false,
    doi=true,
    eprint=false
]{biblatex}
\usepackage{csquotes} 
\usepackage{tocloft}
\usepackage{titlesec}
\usepackage{tocbibind}
\usepackage{appendix}
\usepackage{setspace}
\usepackage{parskip}
\usepackage{footmisc}
\usepackage{booktabs}
\usepackage{graphicx}
\usepackage{tikz}
\usepackage{pgfplots}
    \pgfplotsset{compat=1.17}
\usepackage{epstopdf}
\usepackage{caption}
\usepackage{subcaption}
\usepackage{siunitx,eurosym}
\usepackage{pdfx}
\usepackage{float}
\usepackage{multicol,multirow}
\usepackage{bm}
\usepackage{placeins}
\usepackage{hyperref}
    \hypersetup{final, colorlinks=true, urlcolor=blue, citecolor=red}
\usepackage[noabbrev,capitalise]{cleveref}
\usepackage{authblk}
\newcommand{\calG}{{\mathcal{G}}}
\newcommand{\calH}{{\mathcal{H}}}
\newcommand{\bb}{{{\bm \beta}}}

\renewcommand{\P}{\mathbb{P}}

\newcommand{\R}{\mathbb{R}}
\newcommand{\gstep}{(G_t \rightarrow G_{t+1})}

\newcommand{\ie}{{\em i.e.}, }
\newcommand{\eg}{{\em e.g.}, }
\newcommand{\vrg}[1]{{``#1''}}

\newtheorem{thm}{Theorem}

\usepackage{algorithm}
\usepackage{algpseudocode}

\title{Generating synthetic power grids\\using exponential random graphs models}

\author[1]{Francesco Giacomarra\thanks{This study was carried out within the PNRR research activities of the consortium iNEST funded by the European Union Next-GenerationEU (PNRR, Missione 4 Componente 2, Investimento 1.5 – D.D. 1058 23/06/2022, ECS\_00000043).}}
\author[2]{Gianmarco Bet}
\author[3]{Alessandro Zocca}
\affil[1]{Department of Mathematics, Informatics and Geosciences, University of Trieste}
\affil[2]{Department of Mathematics and Computer Science \vrg{Ulisse Dini}, University of Florence}
\affil[3]{Department of Mathematics, Vrije Universiteit Amsterdam}

\date{\today}

\begin{document}
\maketitle

\begin{abstract}
Synthetic power grids enable secure, real-world energy system simulations and are crucial for algorithm testing, resilience assessment, and policy formulation. We propose a novel method for the generation of synthetic transmission power grids using Exponential Random Graph (ERG) models. Our two main contributions are: (1) the formulation of an ERG model tailored specifically for capturing the topological nuances of power grids, and (2) a general procedure for estimating the parameters of such a model conditioned on working with connected graphs. From a modeling perspective, we identify the edge counts per bus type and $k$-triangles as crucial topological characteristics for synthetic power grid generation. From a technical perspective, we develop a rigorous methodology to estimate the parameters of an ERG constrained to the space of connected graphs. The proposed model is flexible, easy to implement, and successfully captures the desired topological properties of power grids.
\end{abstract}

\section{Introduction}
Power grids are fundamental infrastructures of our modern societies and economies. A thorough understanding of the principles that govern the formation of these networks is crucial to guarantee their reliability at all times. However, network operators release only limited information on transmission power grids due to security concerns. Therefore, there is a substantial lack of real high-quality data.
To address this problem, over the last two decades, synthetic grid generation approaches have been extensively investigated by the research community. The main challenge has been to develop models flexible enough to replicate the very heterogeneous nature and peculiar properties of real power grids. 

We propose a novel approach for the generation of synthetic power grids based on Exponential Random Graph (ERG) models. The main idea of an ERG is to consider a parametric probability density over the space of all graphs and tune its parameters to encode desirable topological properties as soft constraints. From a modeling perspective, ERGs have been quite successful due to their flexibility since they offer a very tractable alternative to the problem of sampling from involved graph subspaces. To the best of our knowledge, this class of models has not been considered yet in the power systems literature. This is probably due to the intrinsic difficulty of generating graphs that are both sparse and connected, both key topological properties of transmission power grids.


The main contributions of this paper are the following: first, we propose an ERG model that captures the main topological properties of real power grids. Second, we give a general procedure to estimate the parameters of a wide class of ERG models with constraints using an algorithm based on Markov chain Monte Carlo with noisy parameters, which we prove to converge to the set of parameters that satisfies the constraints imposed by the ERG model defined before. We present the results that we have obtained with our procedure, showing that the proposed model is flexible and captures the properties of possibly very different power grids while also being simple, easy to implement, and theoretically grounded. We remark that the proposed methodology is rather general and, except for the choices of graph statistics, is in fact not specific to synthetic power grid generation.

In the last two decades, various proposals have been put forward to address the challenge of generating synthetic yet realistic power grids. Some major commonalities among these works can be identified, based on the approach used to tackle the problem.
Similarly to the approach that we propose here, there have been several attempts to use already existing graph models either to directly generate power grid topologies or as a building block for more complex procedures. Examples of models used are the \vrg{small-world model} introduced by Watts and Strogatz in \cite{Watts1998} and refined by Wang \textit{et al.} into the \vrg{RT-nested small-world model} to generate synthetic power grids \cite{Wang2010,Wang2015,Wang2017,Wang2018}. Another model used is the Chung-Lu model \cite{Chung2002a,Chung2002b,Chung2003} used as the building block for the generation procedures presented in \cite{Aksoy2018,Boyaci2022}. Similarly, in~\cite{Scoglio2014} a variation of the Generalized Random Graph model~\cite{Molloy1995} is proposed as a generative model.

Many researchers argue that the geographical attributes of the area and/or the geographical locations of the nodes cannot be disregarded while generating synthetic power grids. Consequently, several models proposed in the literature put particular emphasis on the \textit{spatial embedding} of the synthetic grids. A straightforward way to generate spatially embedded grids that are both connected and sparse is to solve the \vrg{Minimum Spanning Tree} (MST) problem \cite{Nesteril2001} given the desired locations of the nodes and by assigning weights at each possible edge based on some distance/cost function. MSTs are often used as the first step to build the topology of several synthetic grid generation procedures \cite{Schultz2014,Soltan2016,Soltan2017,Rantaniemi2022}.

Other approaches in the literature to obtain spatially embedded synthetic grids rely on clustering of nodes based on the geographical properties of the nodes' locations \cite{Birchfield2017,Birchfield2017b, Espejo2019}. After the cluster identification, different procedures have been proposed to obtain the topologies with the desired properties considering as distinct the edges that connect nodes within the same cluster and those that connect nodes belonging to different clusters. 

The authors in~\cite{Halappanavar2015,Young2018,Young2018b} propose to view the power grids from the perspective of \vrg{network of networks}, \ie analyzing separately the subgraphs with the same voltage level in the grid, which they call fragments, and then the interconnections of these subgraphs as a new graph itself. This hierarchical view led to the development of the so-called \textit{Sustainable Data Evolution Technology} (SDET) tool to create open-access synthetic grid datasets~\cite{Huang2018}. Using this method, new synthetic topologies are generated by reassembling fragments of real grids that were previously anonymized (to avoid disclosure of sensitive information).
The idea of developing an anonymization procedure rather than a completely new generation method is also discussed in \cite{Vargas2021}, where the grid topologies of real grids are left unchanged and only the electrical parameters are randomized to obfuscate sensitive information.

Finally, in~\cite{Khoydar2019} a deep learning method for the generation of synthetic power grid is presented. To our knowledge, this is the first attempt to solve this kind of problem using deep learning techniques. This is because, generally, these approaches require a large amount of input data in order to provide accurate results.

The paper is structured as follows. In ~\cref{sec:topo} we introduce the graph-theoretical framework and discuss the main topological properties of power grids. In \cref{sec:erg}, we introduce the Exponential Random Graph model along with the proposed specifications. The estimation procedure and a theoretical explanation of its convergence are discussed in \cref{sec:parest}. The results obtained using our proposed model specifications are presented in \cref{sec:results}. Lastly, in \cref{sec:conclusions}, we offer final remarks and highlight potential avenues for future research.

\section{Transmission power grids as complex networks}
\label{sec:topo}
Power grids are interconnected networks that deliver electricity from producers to consumers, consisting of nodes called buses connected through links called power lines. We can distinguish two main types of power networks, namely, the distribution network and the transmission network. Distribution networks have shorter power lines (often referred to as distribution power lines) and serve the function of transporting electricity for short distances and low voltage levels. The transmission network is used to transport electricity over long distances working at high voltage levels, having longer power lines (also called transmission power lines). Coherently with other similar works in the literature \cite{Wang2008,Wang2010, Scoglio2014, Schultz2014, Soltan2016, Soltan2017, Birchfield2017, Birchfield2020, Elyas2015, Espejo2019, Huang2018, Sadeghian2020, Aksoy2018,Boyaci2022}, we focus solely on the transmission network,  and in this section we will highlight how to model such a network and its properties in a graph-theoretic framework.

\subsection{Model description and preliminaries}
\label{sub:model}
A high-voltage transmission network can be modeled as a simple, undirected, unweighted graph $G=(V,E)$, where the nodes $V=\{1,\dots,|V|\}$ represent electrical buses and the edges $E\subset V\times V$ represent the transmission lines connecting them. We denote by $n:=|V|$ the number of nodes and by $m:=|E|$ the number of edges. 
Any simple graph $G = (V,E)$ can be equivalently described by its \textit{adjacency matrix} $A = A(G) \in \{0,1\}^{n \times n}$, which is the square symmetric matrix defined as
\begin{equation}\label{eq:adjmat}
    A_{ij} = \begin{cases}
    1 & \mbox{if}\ (i,j) \in E \\
    0 & \mbox{otherwise}.
    \end{cases}
\end{equation}

We define the graph distance $d(i,j)$ between any two nodes $i\neq j$ as the length of the shortest path (in hops) between $i$ and $j$, with $d(i,i)=0$ and $d(i,j)=\infty$ if there are no paths between $i$ and $j$. If $d(i,j)<\infty$ for any pair of nodes $i \neq j$, then the graph is said to be \textit{connected}. The \textit{average path length} or \textit{characteristic path length} is the average length of the shortest path between any two nodes in the graph, i.e.,
\begin{equation}\label{eq:apl}
    \langle \ell \rangle = \frac{1}{n(n-1)} \sum_{i,j \in V} d(i, j).
\end{equation}

For every node $i \in V$, we define its \textit{degree} $k_i = \mathrm{deg}(i) \in \mathbb N$ as the number of nodes adjacent to $i$ in $G$. The degree $k_i$ of node $i\in V$ can be recovered as the sum of the $i$-th row of the adjacency matrix, namely $k_i = \sum_{j=1}^n A_{ij}$. The \textit{average node degree} $\langle k \rangle$ of the graph $G$ is
\begin{equation*}
    \langle k \rangle := \frac{1}{n} \sum_{i=1}^{n} k_i = \frac{2m}{n}.
\end{equation*}
The \textit{degree matrix} $D(G)$ of the graph $G$ is the square diagonal matrix defined as $D(G)=\mathrm{diag}(k_1,\dots,k_n)$. Another equivalent matrix representation of the graph $G$ is given by its \textit{Laplacian matrix} $L=L(G) \in \mathbb{R}^{n \times n}$, which is the square symmetric matrix defined as $L(G):= D(G)-A(G)$, or, equivalently, as
\begin{equation}\label{eq:laplacian}
L_{i,j} := \begin{cases}
    k_i & \mbox{if}\ i = j \\
    -1 & \mbox{if}\ (i,j) \in E \\
    0 & \mbox{otherwise}.
\end{cases}
\end{equation}
The Laplacian matrix is a useful graph-theoretical tool, and its spectrum can be linked to many properties of the corresponding graph, for a detailed overview, see \cite{Chung1997}. In particular, the second-smallest eigenvalue $\lambda_2$ of the Laplacian matrix, known as \textit{algebraic connectivity}, is closely related to the connectivity of the graph itself~\cite{Fiedler1989}. In particular, $\lambda_2 >0$ if and only if $G$ is a connected graph and its magnitude reflects how well connected $G$ is. Spectral graph theory has been shown to be an excellent tool for understanding the redistribution of power flows after line contingencies~\cite{Ronellenfitsch2017,Kaiser2021b,Zocca2021,Guo2021a,Guo2021b} and thus for designing more robust network topologies~\cite{Kaiser2021a,Lan2023}.

A \textit{triangle} is a clique of size $3$, that is, a subgraph consisting of three nodes that have an edge between each pair of nodes. The total number of triangles $t_1(G)$ in a graph $G$ can be calculated using its adjacency matrix as
\begin{equation}
    t_1(G) = \sum_{i=1}^n\sum_{j=1}^n \sum_{l=1}^n A_{ij} A_{jl} A_{li}.    
\end{equation}
More generally, for each $k \geq 1$ one defines \textit{$k$-triangles} as the subgraphs in which the $k$ triangles share the same edge; see \cref{fig:ktriangs} for some examples.
\begin{figure}[!h]
    \centering
    \begin{subfigure}[b]{0.25\textwidth}
    \includegraphics[width=0.99\textwidth]{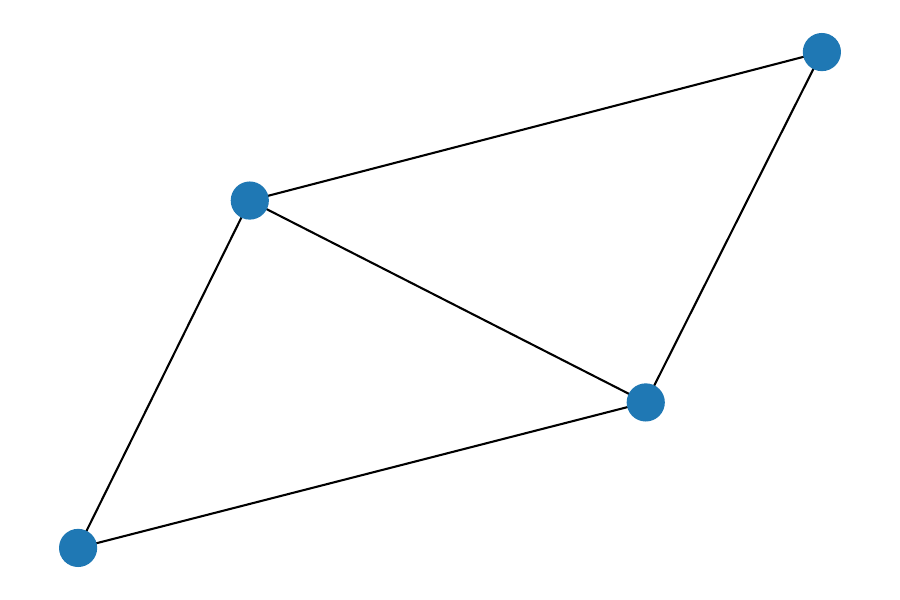}
    \caption{}
    \end{subfigure}
    \begin{subfigure}[b]{0.25\textwidth}
    \includegraphics[width=0.99\textwidth]{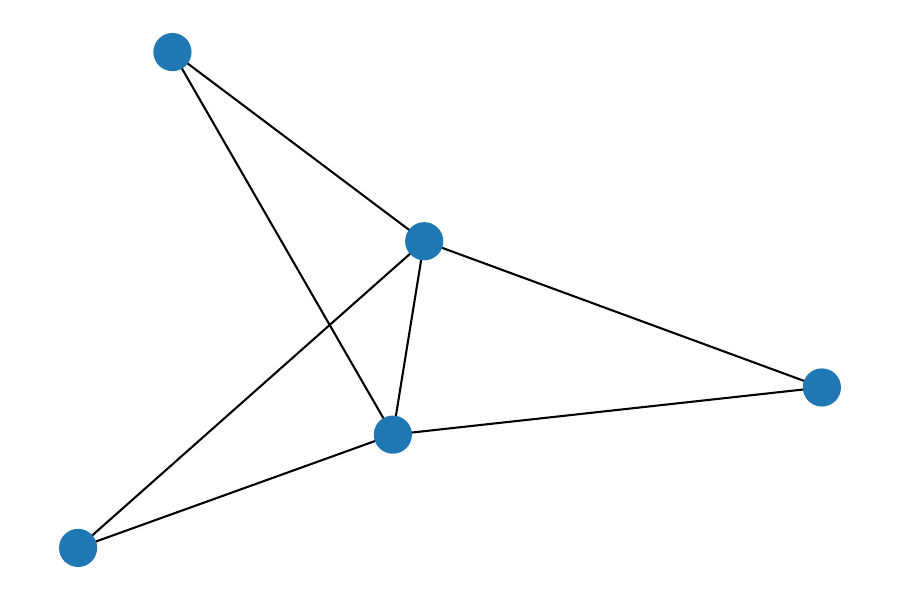}
    \caption{}
    \end{subfigure}
    \begin{subfigure}[b]{0.25\textwidth}
    \includegraphics[width=0.99\textwidth]{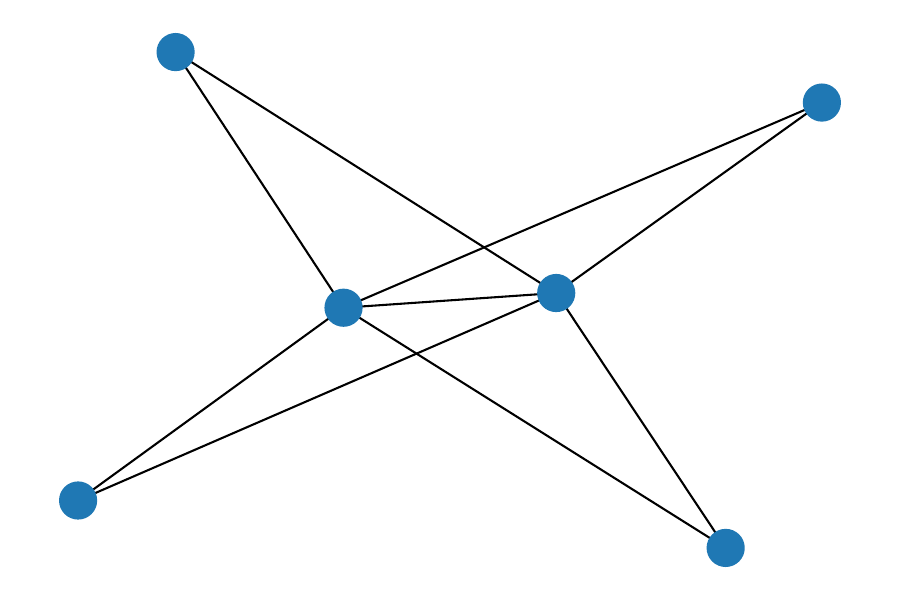}
    \caption{}
    \end{subfigure}
    \caption{Examples of (a) $2-$triangle, (b) $3-$triangle, and (c) $4-$triangle}
    \label{fig:ktriangs}
\end{figure}
We denote the total number of $k$-triangle in a graph $G$ as $t_k(G)$.

The \textit{local clustering coefficient} $C_i$ of node $i \in V$ is defined as the ratio between the number of triangles $\mathcal{T}_G(i)$ to which node $i$ belongs and the maximum number of triangles $t_G(i)$ that could possibly exist between node $i$ and its $k_i$ neighbors. Denoting by $N_i \subseteq V$ the \textit{neighborhood} of $i$, that is the collection of nodes adjacent to $i$ in $G$, the local clustering coefficient can be computed as
\begin{equation}\label{eq:localcdef}
    C_i = \frac{\mathcal{T}_G(i)}{t_G(i)} = \frac{|\{(j,k) \in E ~:~ j,k \in N_i\}|}{k_i(k_i-1)/2} = \begin{cases}\displaystyle \frac{1}{k_i(k_i-1)}\sum_{j=1}^n \sum_{l=1}^n A_{ij}A_{jl}A_{li} 
 & \mbox{if}\ k_i \geq 2,\\
    0  & \mbox{if}\ k_i < 2.
    \end{cases}
\end{equation}
The \textit{average clustering coefficient} $C$ is defined as the average of the local clustering coefficients, i.e.,
\begin{equation}\label{eq:globalc}
    C= \frac{1}{n} \sum_{i=1}^n C_i.
\end{equation}

In this work, we focus purely on the topology of synthetic grid graphs and thus do not consider the electrical properties of the nodes or of the lines. We refer the interested reader to~\cite{Wood2013}. However, it is key to distinguish the nodes based on the function of the corresponding substation, as this information is used in the generative procedure. In the context of power networks, we generally distinguish three types of nodes:
\begin{itemize}
    \item \textit{generator nodes}, which represent the network components where the electricity is produced, e.g., conventional power plants, wind parks, or solar panels.
    \item \textit{load nodes}, which represent the network components where electricity is consumed, for example, industrial districts, residential areas or distribution network feeders.
    \item \textit{interconnections nodes}, which represent intermediate substations or transformers.
\end{itemize}

In \cref{fig:tops} we show the topologies of two grids made available in \cite{Sogol2019}. The two grids, namely the \texttt{118 IEEE} and the \texttt{300 IEEE}, have respectively 118 and 300 nodes. For each grid, in addition to the topology, we highlight the bus-type assignment. More details on the grids used in our work are given in \cref{app:first}.

\begin{figure}[!ht]
    \centering
    \begin{subfigure}[b]{0.485\textwidth}
    \includegraphics[width=0.99\textwidth]{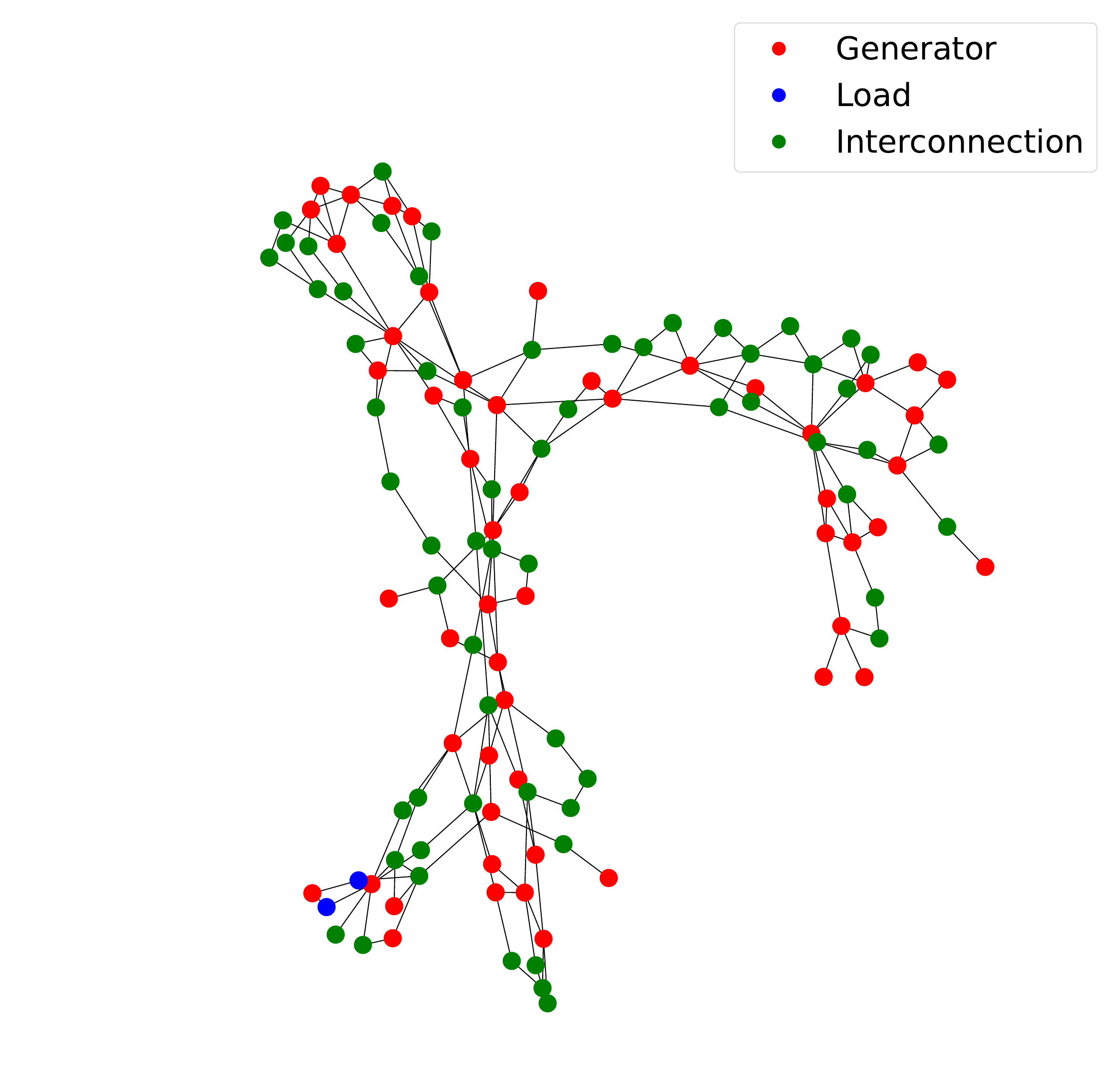}
    \caption{}\label{fig:orig}
    \end{subfigure}
    \begin{subfigure}[b]{0.485\textwidth}
    \includegraphics[width=0.99\textwidth]{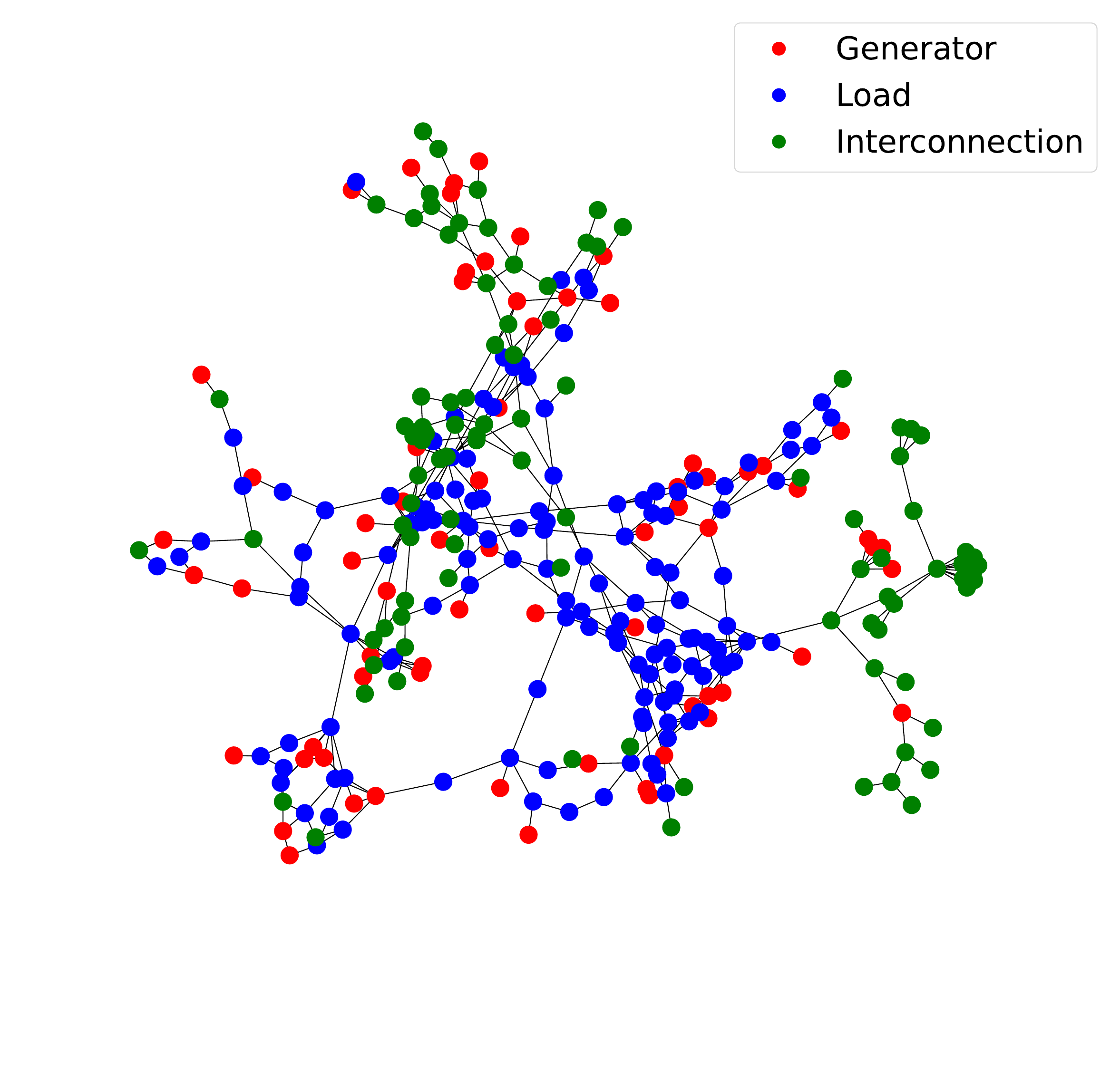}
    \caption{}\label{fig:gener}
    \end{subfigure}
    \caption{Visualization of the topology of the \texttt{118 IEEE} network (a), and the topology of the \texttt{300 IEEE} network (b). The node colors reflect the bus-type assignment. Since no geometric embedding was given, the nodes position were determined using a Force-directed graph drawing algorithm \cite{Fruchterman1991}.}
    \label{fig:tops}
\end{figure}
\FloatBarrier

\subsection{Topological properties of power grids}
\label{sub:properties}
A large body of literature, see, \eg \cite{Pagani2013, Birchfield2017, Birchfield2017b, Wang2018}, examined various topological properties of real power grids, which are instrumental in assessing the realism of synthetic grids. After an analysis of the available grids described and made available in \cite{Sogol2019}, we now revisit and discuss some of these key properties.

\begin{itemize}
    \item \textbf{Connectivity.} Power grids are fully connected graphs, which means there exists a path between any two nodes, except for rare emergency situations. This is due to the requirements for the reliability and security of the power grid system. Indeed, under normal operating conditions, electric power must be able to flow from any point on the grid to any other point. Most real power grids have an even stronger connectivity property, namely they are \textit{2-edge-connected graphs}. This property is, in fact, part of the standard contingency analysis known as \textit{$N-1$ security}, which ensures that the power grid can withstand the failure of any single component (line, transformer, generator, etc.) without losing the ability to supply power to all the remaining loads.
    
    \item \textbf{Average node degree and sparsity.} It has been shown that the \textit{average node degree} $\langle k \rangle$ in real power grids oscillates between $2$ and $3$ regardless of the network size~\cite{Wang2008,Pagani2013}. Therefore, power grids are sparse graphs, which means that the number of edges is of the same order of magnitude as the number of nodes, informally $|E| = \mathcal{O}(n)$. This can be intuitively explained by the high costs of building and maintaining transmission lines, as well as practical engineering constraints, such as the avoidance of transmission line crossings.

    \item \textbf{Clustering coefficient and total number of triangles.} The average clustering coefficient $C$ has been empirically observed to be much higher than that of other types of sparse graphs, which means that power grids tend to have many more triangles than other sparse graphs~\cite{Wang2018}. This could be due to the fact that a single line's removal should not disconnect any node from the others, and triangles are the simplest subgraph structure that allows for this property.
    We further investigate the clustering properties by looking at the number of $k$-triangles of the power grids in the publicly available dataset~\cite{Sogol2019}. Numerical evidence suggests that the number of triangles in power grids grows linearly with the number of edges; see~\cref{fig:edgetriang}. This explains the observed high average clustering coefficients (cf.~\cref{tab:allgrids} in the Appendix). On the other hand, the number of $2-$triangles (and consequently that of any $k-$triangle with $k\geq 2$) is roughly constant, in fact often very close to zero, and does not grow with the size of the network.    
    \begin{figure}[!h]
        \centering
        \includegraphics[width=0.9\textwidth]{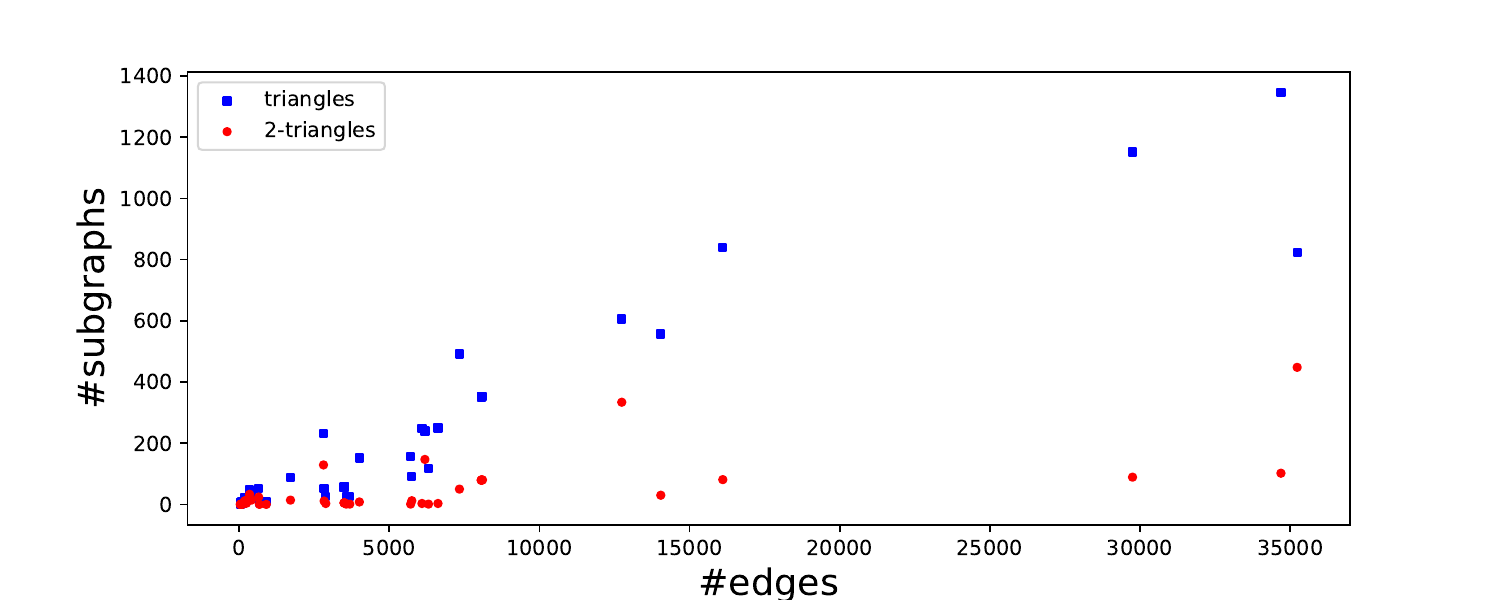}
        \caption{
        The number $t_1$ of triangles (in blue) and the number $t_2$ of $2-$triangles (in red) reported against the number of edges for various test networks (see \cref{tab:allgrids}) }
        \label{fig:edgetriang}
    \end{figure}

    \item \textbf{Bus-type assignment.} As reported in~\cite{Wang2015b}, in a typical power grid, 20-40\% are generation buses, 40-60\% are load buses, and about 20\% are interconnection buses. The authors also suggest that there exists a correlation between bus-type assignment and several network topology metrics.

    \item \textbf{Average degrees per node type} The degree distribution of the nodes in power grids has been shown to be different for the nodes of different types \cite{Wang2010}, motivating the introduction and the study of the average node degree per bus-type $\langle k_P \rangle $, $\langle k_L\rangle $, $\langle k_I \rangle $, where
    \begin{equation}\label{eq:avgtype}
        \langle k_{a} \rangle := \frac{1}{|a|} \sum_{i \in a} k_i, \quad a \in \{P,L,I\}.
    \end{equation}

    \item \textbf{Algebraic connectivity.} It has been shown \cite{Wang2010} that algebraic connectivity scales as a power of the network size $n$, \ie $\lambda \sim n^p$ with $p$ said to be in the range $ [ -1.376, -1.06]$. 

    \item \textbf{Average shortest path length.} Assuming power grids follow the \vrg{small-world} property as claimed by Albert and Barabasi \cite{Albert2002} (but there is no consensus in the literature, cf.~\cite{Pagani2013}), then the APL should grow proportionally to $\ln (m)/ \ln (\langle k \rangle)$.

    \item \textbf{Graph diameter}: the graph diameter $d_{\max}$ of real power grids has been shown to scale roughly as $\sqrt{n}$ \cite{Cotilla-Sanchez2012,Cloteaux2013,Aksoy2018,Young2018,Espejo2019}. 
\end{itemize}

All these quantities for the considered grids are reported in \cref{tab:allgrids} in \cref{app:first}.

\section{Synthetic grid generation using an Exponential Random Graph Model}\label{sec:erg}

In this section, we briefly discuss the probabilistic details of the Exponential Random Graph (ERG) model, as well as the rationale behind the proposed model formulation, while also stating how graphs with properties similar to the ones exhibited by the power grids have been modeled in the ERG literature.

\subsection{General ERG model formulation}
Let $\calG_n$ be the collection of all undirected, unweighted, simple graphs on $n$ nodes. We denote by $\mathbf{G}$ a random variable that takes values in $\calG_n$. An \textit{Exponential Random Graph} model \cite{Fronczak2012} is a probability distribution on  $\calG_n$ of the form
\begin{equation} \label{eq:ergm}
    \mathbb P_\bb(\mathbf{G}=G) = \frac{\exp(\calH_{\bb}(G))}{Z_\bb} = \frac{\exp \left(\sum_{i=1}^r \beta_i x_i(G)\right)}{Z_\bb}, \qquad G \in \calG_n,
\end{equation}
where $\bb \in \R^r$ is a vector of parameters, $x_i:\calG_n\mapsto\mathbb R$ are graph statistics of interest, also known as the \textit{observables of the model}, and $Z_\bb:=\sum_{G \in \calG_n} e^{\mathcal{H}_{\bb}(G)}$ is a normalizing constant. The function $\calH_{\bb}: \calG_n \rightarrow \mathbb{R}$ defined by $\calH_{\bb}(G) := \sum_{i=1}^r \beta_i x_i(G) $ is commonly called the \textit{Hamiltonian} of the model. 

By definition, the density \eqref{eq:ergm} is positive for all graphs $G\in\calG_n$. Intuitively, graphs $G$ with large values of $x_i(G)$ become less (resp.~more) likely if $\beta_i<0$ (resp.~$\beta_i>0$). By carefully choosing the parameters $\bb$, the expected values $(\mathbb E[x_i(\mathbf G)])_{i=1}^r$ of all the observables across the entire collection $\calG_n$ can be tuned. In fact, the ERG density in \eqref{eq:ergm} is the unique probability density $\P(G)$ over $\calG_n$ that maximizes the Shannon entropy $S(\P)$ defined as
\begin{equation}\label{eq:shannon}
    S = -\sum_{G \in \calG_n} \P(G) \ln \P(G),
\end{equation}
subject to the constraints
\begin{equation}\label{eq:mean_constraints}
    \mathbb E[x_i(\mathbf G)] = \overline{x}_{i},\qquad i=1,\ldots, r,
\end{equation}
where $\overline{x}_{i}$ is the desired average value of observable $i$. In other words, the distribution \eqref{eq:ergm} does not contain further structured information other than the constraints \eqref{eq:mean_constraints}.

Assuming the target values $\overline{x}_i$, $i=1,\dots,r$, of the chosen $r$ observables are given, one can tune the parameters $\bb$ by solving the following system of equations, obtained by rewriting the constraints~\eqref{eq:mean_constraints} using~\eqref{eq:ergm}
\begin{equation}\label{eq:params}
    \overline{x}_i = \mathbb E[x_i(\mathbf G)] = \frac{1}{Z_\bb}\sum_{G \in \calG_n} x_i(G) \mathrm{e}^{\mathcal{H}_{\bb}(G)} = \frac{1}{Z_\bb} \frac{\partial}{\partial \beta_i}\sum_{G \in \calG_n} e^{\mathcal{H}_{\bb}(G)} 
     = \frac{1}{Z_\bb}\frac{\partial Z_\bb}{\partial \beta_i} = \frac{\partial F_\bb}{\partial \beta_i}, \qquad i=1,\dots,r,
\end{equation}
where $F_\bb := \ln Z_\bb$ is the so-called \textit{free-energy} of the model. However, this strategy requires having a closed-form expression for the partition function $Z_\bb$ or of the free energy $F_\bb$, which crucially depends on the choice for the Hamiltonian \eqref{eq:hamf} and is not available in most cases. 

\subsection{Proposed ERG models}
In this paper, we propose to use a synthetic grid procedure that samples random graphs using the probability distribution specified by an ERG model of the form~\eqref{eq:ergm}. 
Our choice of observables is driven by the considerations made in \cref{sub:properties}, and by a careful analysis of publicly available data from real grids collected and described in \cite{Sogol2019}.

To fully introduce the model, we first need some preliminary definitions and additional notation. Consider a generic undirected graph $G=(V,E) \in \calG_n$ with a fixed number $n$ of nodes. Consistent with~\cref{sub:model}, we further assume that each node is either a generator node, a load node, or an interconnection node. Denote by $P, L, I \subset V$ the three corresponding subsets of generator, load, and interconnection nodes, respectively, so that $V = P\cup L \cup I$. 
Given two types of nodes $a, b\in\{P,L,I\}$, we denote by $E_{ab}(G) \subset E$ the subset of edges in the graph $G$ that connect a node of type $a$ to one of type $b$. Since we are working with an undirected graph, we have $E_{ab}(G) = E_{ba}(G)$. In this way, we obtain a partition of the edge set $E$ as $E = E_{PP}(G) \cup E_{PL}(G) \cup E_{PI}(G) \cup E_{LL}(G) \cup E_{LI}(G) \cup E_{II}(G)$.

For each pair of node types $a, b\in\{P,L,I\}$, we denote the cardinality of the corresponding edge subset by $e_{ab}(G):=|E_{ab}(G)|$. Including these six edge counts as observables in the Hamiltonian allows us to simultaneously tune the average edge density and the average degree of the typical vertex of each bus type. However, as we will illustrate later in \cref{sub:ablation}, the ERG model obtained by considering only these six edge counts as graph observables does not perform well, as the sampled graphs do not have many of the desired properties mentioned in~\cref{sub:properties}. In particular, this simple ERG model is unable to realistically replicate the clustering structure of power grids.

To overcome this limitation and capture the clustering properties of real power grids, we thus consider a more involved Hamiltonian that also includes the number of triangles $t_1$, as previously suggested by \cite{Strauss1986}. Recall from~\cref{sub:model} that the average clustering coefficient is defined precisely using the number of triangles. However, the resulting Hamiltonian has been proven to lead to degeneracy, \ie the resulting ERG density assigns most of the probability mass either on (nearly) fully connected graphs or only on random bipartite graphs, see \cite{Chatterjee2013}, depending on the sign of the parameter associated to the triangle count.

Snijders \textit{et al.}~\cite{Snijders2006} have introduced a new class of models that exhibit the desired clustering coefficient and are not prone to degeneracy: the main idea is to include a more elaborate function of the number $t_k$ of $k$-triangles with $k\geq 1$ rather than just the number $t_1$ of triangles. 
Specifically, we consider the so-called \textit{alternating $k$-triangles statistic} of a graph $G$ introduced in~\cite{Snijders2006}, which is defined as
\begin{equation}\label{eq:altk}
    u_\zeta (G) = 3t_1 (G) - \frac{t_2(G)}{\zeta} + \frac{t_3(G)}{\zeta^2} - \frac{t_4(G)}{\zeta^3} + \dots + (-1)^{n-3}\frac{t_{n-2}(G)}{\zeta^{n-3}}.
\end{equation}
Intuitively, the alternating signs and decreasing weights in $u_\zeta(G)$ compensate each other, leading to highly clustered graphs that are not (nearly) fully connected, or nearly empty.

In order to develop a new ERG model to generate synthetic power grids, some practical considerations are in order. Indeed, for most of the power grid topologies that are publicly available, the number of triangles $t_1$ increases linearly in the number of edges, while the number $t_2$ of $2-$triangles is essentially constant, see \cref{fig:edgetriang}. More generally, we have found in our analysis that the number of $k-$triangles is close to $0$ for all power grids when $k>2$. For this reason, rather than the involved alternating $k$-triangle statistic $u_\zeta (G)$, we chose to include in the ERG Hamiltonian two independent terms for the number of $1$-triangles and $2$-triangles (with two independent parameters $\beta_{1t}$ and $\beta_{2t}$).
The rationale behind this choice is the following: by ignoring all $k-$triangles counts for $k > 2$, we enormously reduced the computational effort required by the model, while still maintaining the mitigating effect of the alternating signs of the parameters by imposing $\beta_{1t}\cdot \beta_{2t}\leq 0$, following \cite{Snijders2006}. The resulting Hamiltonian thus has $r=8$ terms and reads
\begin{equation}\label{eq:hamf}
    \mathcal{H}_{\bb}(G) = \beta_{PP}e_{PP}(G) + \beta_{PL}e_{PL}(G) + \beta_{PI}e_{PI} + \beta_{LL}e_{LL}(G) +  \beta_{LI}e_{LI}(G) + \beta_{II}e_{II}(G) + \beta_{1t} t_1(G) + \beta_{2t} t_2(G).
\end{equation}
%

This proposed ERG model assumes that the target values for these eight observables are known. To work with realistic target values, in the rest of the paper we adopt the following strategy: we consider a publicly available power grid $G_0$ with $n$ nodes as a reference graph and use it to compute the target values as 
\begin{equation}
\overline{x}_i = x_i(G_0), \qquad i=1,\dots,r.
\end{equation}

It is worth mentioning that, given the fact that the observables related to the edges in the Hamiltonian are negative, we expect the number of edges in the graphs generated with the model \eqref{eq:hamf} to be on average directly proportional to the number of nodes in the graph, which is consistent with what has been observed for real power grids, as stated in \cref{sec:topo}.

The addition of the $1$-triangle and $2$-triangle terms in the Hamiltonian results in a partition function $Z_\bb$ that has no closed-form expression, hence making it impossible to solve the systems of equations~\eqref{eq:params} algebraically. Moreover, the collection $\calG_n$ on which we have defined all ERG models so far contains many disconnected graphs, which are not of interest when modeling power grid topologies, as we argued in~\cref{sub:properties}. However, it is not possible to explicitly take this requirement into account in the Hamiltonian, as there is no simple algebraic expression of the adjacency matrix that can capture the connectivity of the graph. Even if we were able to find a suitable proxy for connectivity to add as observable to the Hamiltonian, one should not forget that the ERG density imposes only \textit{soft constraints} based on the observables. However, we want the connectedness of the graph to be a \textit{hard constraint}, as the goal is to sample \textit{connected} synthetic power grids. To accommodate this, we restrict the ERG model to the subsets of the \textit{connected graphs} with $n$ nodes, which we denote as $\calG_{n,\text{conn}} \subset \calG_n$. This is equivalent to sampling from the same ERG model in \eqref{eq:hamf} but \textit{conditional on} on the graph being connected. 

The density of the ERG model obtained when restricting to a general subset $\mathcal{S} \subset \calG_{n}$ is
\begin{equation} \label{eq:ergm_const}
    \mathbb P_{\bb_{\mathcal{S}}}(\mathbf{G}=G) = \frac{\exp(\calH_{\bb_{\mathcal{S}}}(G))}{Z_{\bb_{\mathcal{S}}}} = \frac{\exp \left(\sum_{i=1}^r \beta_i x_i(G)\right)}{Z_{\bb_{\mathcal{S}}}}, \qquad G \in \mathcal{S},
\end{equation}
where  $Z_{\bb_{\mathcal{S}}}:=\sum_{G \in \mathcal{S}} e^{\mathcal{H}_{\bb}(G)}$. We denote by $\bb_{\mathcal{S}}=(\beta_1,\dots,\beta_r)$ the new different set of parameters that the ERG model needs to express the target average for each graph observable. In fact, when we restrict ourselves to a general subset $\mathcal{S} \subset \calG_{n}$, the system of equations that determine the parameters \eqref{eq:params} also changes. In particular, we have
\begin{equation}\label{eq:params_const}
    \overline{x}_i = \mathbb E[x_i(\mathbf G)] = \frac{1}{Z_{\bb_{\mathcal{S}}}}\sum_{G \in \mathcal{S}} x_i(G) \mathrm{e}^{\mathcal{H}_{\bb}(G)} = \frac{1}{{Z_{\bb_{\mathcal{S}}}}} \frac{\partial}{\partial \beta_i}\sum_{G \in \mathcal{S}} e^{\mathcal{H}_{\bb}(G)} 
     = \frac{1}{Z_\bb}\frac{\partial {Z_{\bb_{\mathcal{S}}}}}{\partial \beta_i} = \frac{\partial F_{\bb_{\mathcal{S}}}}{\partial \beta_i}, \qquad i=1,\dots,r.
\end{equation}

In the special case, we are considering where $\mathcal{S} = \calG_{n, \text{conn}}$ and the Hamiltonian as in \eqref{eq:hamf}, it is not possible to retrieve analytically $Z_{\bb_{\mathcal{S}}}$. Therefore, a numerical estimate is required to obtain a set of parameters that satisfy \eqref{eq:params_const}. 

\section{Parameter estimation using the Equilibrium Expectation algorithm}
\label{sec:parest}
In this section, we introduce a new method that we used to estimate the parameters $\bb$ that satisfy \eqref{eq:params_const} for an ERG model formulation that includes a hard constraint as in \eqref{eq:ergm_const}.
Before doing so, we briefly discuss the literature on ERG parameter estimation and relate it to our scheme.

\subsection{Methods for ERG parameters estimation}
Parameter estimation for ERG models in the presence of an intractable partition function is still an open problem \cite{Fronczak2012}. The use of mean-fields techniques gives unreliable results in related models such as spin glasses \cite{Talagrand2003}, and has been shown to work for ERG models only for specific values of the parameters that make them almost Erd\H{o}s-R\'enyi models \cite{Chatterjee2010,Chatterjee2013}. 

Many parameter estimation approaches are based on Markov-Chain Monte Carlo (MCMC) methods. MCMC methods are widely used to sample from ERG models \cite{Snijders2002}, and more generally from any probability distribution $\pi$. The Metropolis-Hastings (MH) algorithm \cite{Metropolis1953} is an MCMC scheme that is particularly suitable when the target probability distribution $\pi$ is known up to a constant factor, such as in the case of a Gibbs distribution or the ERG model~\eqref{eq:ergm}. 

The MH algorithm produces a sequence of samples from a Markov chain $\{X_t\}_{t \in \mathbb{N}}$ with specific transition probabilities that are obtained as the results of two steps, a proposal step and a subsequent acceptance step, each characterized by a different distribution. The proposal distribution $\mathbb{T}$ specifies the conditional probability $\mathbb{T}(s,s')$ of going from state $s$ to state $s'$; the acceptance distribution $\mathbb{A}$ specifies the probability $\mathbb{A}(s,s')$ of accepting the proposed state $s'$ if the chain previously resided in $s$. The transition probability can thus be rewritten as
\begin{equation}\label{eq:propacc}
    \P_\bb(s,s') = \mathbb{T}(s,s') \mathbb{A}(s,s').
\end{equation}

We will now specify the MH algorithm that we consider to estimate the parameters of the ERG model. Consider a general ERG model of the form \eqref{eq:ergm} defined over a state space $\mathcal{S} \subseteq \calG_{n}$. The target probability density $\pi$ we want to sample from is the ERG density $\P_{\bb}$ given in~\eqref{eq:ergm}.

In the classical MH algorithm for ERG models \cite{Snijders2002}, only moves that prescribe the addition or removal of a single edge are allowed. In terms of the proposal distribution, this means that for every pair of graphs $G_t=(V,E_t),G_{t+1}=(V,E_{t+1}) \in \mathcal{S}$
\begin{equation}\label{eq:transmh}
   \mathbb{T}(G_t,G_{t+1})>0 \quad \Longleftrightarrow \quad |E_t \bigtriangleup E_{t+1} |\leq 1.
\end{equation}
We consider the following simple proposal distribution $\mathbb{T}(G_t, G_{t+1})$. We randomly choose a pair of nodes uniformly $(i,j) \in V \times V$ and if $(i,j) \not\in E_t$, then we add the corresponding edge, obtaining a new set of edges $E_{t+1} = E_t \cup \{(i,j)\}$. Otherwise, we remove the selected edge and obtain a new graph with edge set $E_{t+1} = E_t \setminus \{(i,j)\}$. It is clear that this proposal distribution satisfies condition~\eqref{eq:transmh}.

Once a move is proposed, the acceptance probability is calculated using the desired target distribution \eqref{eq:ergm} by:
\begin{equation}\label{eq:accrulemh}
    \mathbb{A}_\bb(G_t,G_{t+1}) = \min \Big\{1, \frac{\P_{\bb}(G_{t+1})}{\P_{\bb}(G_t)} \Big\},
\end{equation}
with $\P_{\bb}$ being the ERG probability distribution defined in \eqref{eq:ergm}. Note that computing the acceptance probability does not require knowledge of the partition function, since 
\begin{equation}\label{eq:accmh2}
    \frac{\P_{\bb}(G_{t+1})}{\P_{\bb}(G_t)} =\frac{e^{\calH_{\bb}(G_{t+1})}}{e^{\calH_{\bb}(G_t)}},
\end{equation}
and, thus, it is possible to simplify the expression for the acceptance probability into
\begin{equation}\label{eq:ergacc}
    \mathbb{A}_\bb(G_t,G_{t+1})  = \min \Big\{1, e^{\calH_{\bb}(G_{t+1})-\calH_{\bb}(G_t)} \Big\}.
\end{equation}
We denote by $\{X^{(\bb)}_t\}_{t \in \mathbb{N}}$ the Markov chain on the state space $\mathcal{S}$ defined by the MH algorithm. Note that the steady-state distribution $\pi$ of $\{X^{(\bb)}_t\}_{t \in \mathbb{N}}$ exists and is unique since the chain is reversible. Furthermore, using \eqref{eq:ergacc} guarantees that $\pi$ corresponds to the ERG density \eqref{eq:ergm}.

A disadvantage of MCMC methods is that the chain must be close to stationarity to produce samples from the desired ERG distribution $\P_{\bb}$. In practical terms, this means that the chain needs to run for a sufficiently large number of steps, commonly referred to as the \textit{mixing time} of the chain, so that the empirical distribution of the chain is close to its steady state distribution.

One of the most widely used parameter estimation approaches based on MCMC is the so-called \textit{Markov Chain Monte-Carlo} \textit{Maximum Likelihood} introduced by Geyer~\cite{Geyer1991}. The idea behind this procedure is to use samples from an MCMC as the one defined above, with an arbitrarily chosen set of parameters $\bb_{0}$  to approximate the log-likelihood of the distribution $\mathcal{L}_{\bb}$ with an empirical one $\mathcal{L}_{N,\bb_{0}}$, where $N$ is the number of samples drawn from the steady state distribution of the chain, to retrieve the set of parameters that satisfy \eqref{eq:params}. This method theoretically guarantees the asymptotic convergence of $\mathcal{L}_{N,\bb_{0}}$ to $\mathcal{L}_{\bb}$ when the number of samples $N$ approaches infinity. However, convergence is slow if $\bb_0$ is too far from the target one. If this happens, the literature suggests reiterating the procedure many times using the endpoint of the previous iteration as the starting parameter for the next iteration. 
Although coming with asymptotic convergence guarantees, the Maximum Likelihood estimation through MCMC can be computationally unfeasible in our context, especially when the number of nodes is too large. At each iteration of the method with different parameters, a large number of samples must be taken from the steady-state distribution of the chain after it has reached the mixing time. This could be impractical for large $n$. For ERG models for dense graphs, the mixing time has been shown to be of the order $\mathcal{O}(n^2\log{n})$ \cite{Bhamidi2011}, however, currently there are no similar results for ERG models on sparse graphs, \ie ERG densities whose ensemble has an average number of edges $\langle m \rangle$ of the order $\mathcal{O}(n)$.

From a Bayesian perspective, there are many possible methods that have been proposed and used in the literature to tackle this issue (see, e.g.~\cite{Park2018}), but they are known to scale poorly as the number of nodes increases. 
Nevertheless, the Markov Chain Monte Carlo Maximum Likelihood method is widely used in the ERG community, being also the main method used in popular estimation libraries for ERGMs such as the \texttt{ERGM} package \cite{Hunter2008, krivitsky2022} for the \textbf{R} software \cite{Rcore2022}. This package estimates the parameters of the models using Maximum Likelihood estimation, which is approximated using either the MCMCMLE, which we described before, or the Maximum Pseudo-Likelihood estimation (often referred to as MPLE, first introduced in \cite{Besag1975}), or combination of the two.
However, using any of these approaches while working in the space of connected graphs $\calG_{n,\text{conn}}$ is highly nontrivial. Moreover, enforcing sparsity in the graphs while also imposing the hard constraint of connectivity could be computationally unfeasible for the aforementioned methods, motivating the need for new procedures for this specific problem.

\subsection{Estimation algorithm for constrained chains}

In view of the considerations above, we decided to resort to a variation of the \textit{Equilibrium Expectation} (EE) algorithm \cite{Byshkin2018} to estimate the parameters of our ERG model. The EE algorithm uses a modified Metropolis-Hastings Markov-Chain Monte Carlo model, as defined in \eqref{eq:propacc}, in which the parameters $\bb$ of the ERG model are adjusted dynamically. In contrast to other MCMC methods, it is based on the properties of the chain at equilibrium rather than on drawing a large number of samples from the chain itself, thus decreasing the computational burden.

The initial state $X_0 = G_0$ of the Markov chain associated with the EE algorithm must be drawn from the desired steady-state distribution to exploit the properties of the chain at equilibrium. In our setting, this requirement is satisfied by taking as an initial state the graph corresponding to the power grid used as a reference for the Hamiltonian \eqref{eq:hamf}.

At step $t$ of the chain, the parameters $\bb$ are updated according to some update rule if $t \equiv 0 \pmod{\theta}$, with $\theta$ being a user-defined variable that controls how often the update occurs. The update rule proposed in \cite{Byshkin2018} works as follows: Let $ \bb_{t} =(\beta^{t}_1, \dots, \beta^{t}_r)$ be the parameters associated with the observables $x_1(G),x_2(G),\dots,x_r(G)$ at time $t$. At step $t+1$, the parameters are updated as 
\begin{equation}\label{eq:uprule}
    \beta_{i}^{t+1} = \begin{cases}
    \beta_{i}^{t} + \alpha \cdot \max (|\beta_{i}^{t}|,c) \cdot \mathrm{sign}[x_i(G_0) - x_i(G_t)] & \mbox{if}\ t \equiv 0 \pmod{\theta} \\
    \beta_{i}^{t} & \mbox{otherwise},
\end{cases}
\end{equation}
where $\alpha>0$ is the learning rate and $c>0$ is a control parameter that ensures that the algorithm does not get stuck when $|\beta_{i}^{t}|\ll1$.

This approach alone does not allow for sampling from a constrained space such as $\calG_{n,\text{conn}}$. To overcome this, we combine the EE algorithm with the modified Metropolis-Hastings-type algorithm proposed by Grey \textit{et al.} in \cite{Grey2019}, designed specifically to sample from constrained sets of graphs. 

Consider a nonempty connected subset $\mathcal{S} \subset \calG_n$ and a Hamiltonian $H_\bb$ with $r$ observables. We will prove that if the Markov chain $\{X_t\}_{t \in \mathbb{N}}$ on $\mathcal{S}$ defined by the EE algorithm (i) satisfies \eqref{eq:transmh} and
\begin{equation}\label{eq:transmod}
     \mathbb{T}(G,G')=0 \qquad \forall \, G \in \mathcal{S}, \forall \, G' \notin \mathcal{S},
\end{equation}
and (ii) uses the acceptance rule as in \eqref{eq:ergacc}, then as $t \to \infty$ the parameter process ${\bb}^t$ converges to the solution ${\bb}_{\mathcal{S}}$ of the equations
\begin{align}\label{eq:params2}
    \overline{x}_i(G_0) &= \frac{1}{Z}\sum_{G \in \mathcal{S}} x_i(G) \mathrm{e}^{\mathcal{H}_{\bb_{\mathcal{S}}}(G)}, \quad i\in\{1,\dots,r\}.
\end{align}

In our case, we take $\mathcal{S}$ to be the set of connected graphs of size $n$, that is, $\mathcal{S} = \calG_{n,\text{conn}}$. For~\eqref{eq:transmod} to hold, the Markov chain should be defined to have nonzero transition probabilities only between connected graphs. To this end, we choose the following proposal density $\mathbb{T}(G_t,G_{t+1})$: either an edge chosen uniformly from the set of missing edges $E \setminus E_t $ is added, or an edge chosen uniformly among the edges in $E_t$ that can be removed without disconnecting the graph is removed. The pseudocode is given in \cref{alg:final}. 

\begin{algorithm}[!ht]
\caption{Equilibrium Expectation for a constrained chain}\label{alg:final}
\begin{algorithmic}
\State \textbf{Input:} the starting graph $G_0 \in \mathcal{S}$; $\theta$; $\bb_0$ ; $T$; $\alpha$; $c$ 
\State Set $G_0:=G_0$ and $t=0$
\For{$t=1, \dots, T$} 
\State Sample a random pair $i,j, i \neq j$ uniformly at random from $1,\dots,n$ 
\If{$(i,j) \in E_{t-1}$}
    \State Remove the edge $(i,j)$ and consider a new edge set $E_t = E_{t-1}\setminus (i,j)$
    \If{$G_t \in \mathcal{S}$}
        \State Accept $G_t$ with probability $\min \Big\{1,\P_{\bb}(G_t)/\P_{\bb}(G_{t-1}) \Big\}$ 
    \Else
        \State Reject $G_t$
        \EndIf
\Else
    \State Add the edge $(i,j)$ and consider a new edge set $E_t = E_{t-1} \cup (i,j)$
    \State Accept $G_t$ with probability $\min \Big\{1,\P_{\bb}(G_t)/\P_{\bb}(G_{t-1}) \Big\}$
\EndIf
\State Update $\bb_t$ according to rule \eqref{eq:uprule}
\EndFor
\State Compute the estimate $\displaystyle \bar{\bb} := \bar{\bb}(T) =  \frac{1}{T} \sum_{t=1}^{T} \bb_{t}$
\State \textbf{Output: $\bar{\bb}$}
\end{algorithmic}
\end{algorithm}

This algorithm takes as input a starting graph $G_0$ for the chain, the number of steps $\theta$ of the Markov chain after which the parameters are updated, the starting parameters $\bb_{0}$, the hyperparameters $\alpha$ and $c$ used for the update rule \eqref{eq:uprule} and the maximum number of iterations $T$.
The EE algorithm works as follows: each proposed move is either the addition of an edge or the removal of an existing edge which does not disconnect the graph. In either case, the move is accepted according to the standard Metropolis-Hastings acceptance rule.

After every $\theta$ transitions of $\{X_t\}_{t \in \mathbb{N}}$, each parameter $\beta_i$ is updated simultaneously using the update rule \eqref{eq:uprule}. The limit of the sequence $\bb_t$ generated by the EE algorithm does not depend on the initial values $\bb_0$. In fact, later in~\cref{thm:main_thm} we will show that the limit is unique. However, the choice of $\bb_0$ in general affects the convergence speed of the algorithm, but we did not investigate this numerically. The starting values of $\bb_0$ could be obtained by using, for instance, the Contrastive Divergence method, as suggested in \cite{Byshkin2018}.


The following theorem summarizes the results concerning the convergence of the method.
\begin{thm}\label{thm:ee}
Consider the coupled stochastic processes $(X_t, \bb_t)_{t \in \mathbb{N}}$ returned by \cref{alg:final}. Let $\{X_t\}_{t \in \mathbb{N}}$ be a Markov chain with transition probabilities defined as in \eqref{eq:propacc} on a non-empty and connected subset of graphs $\mathcal{S} \subseteq \calG_n$, with a proposal $\mathbb{T}$ as in \eqref{eq:transmh} and acceptance probability $\mathbb{A}_{\bb_t}$ as in \eqref{eq:ergacc}, and let $\{ \bb_t \}_{t \geq 0}$ be the $r$-dimensional stochastic process describing the evolution of the ERG parameters over time using the update rule \eqref{eq:uprule}. 
If the starting point $X_0$ is drawn from the steady-state distribution of $\{X_t\}_{t \in \mathbb{N}}$ and the learning rate $\alpha$ in the update rule is small enough, then for any $\bb_0$
\begin{equation}\label{eq:est_lim}
    \lim_{T \rightarrow \infty} \bar{\bb}(T) = \lim_{T \rightarrow \infty} \frac{1}{T} \sum_{t = 1}^{T} \bb_{t} = \bb_{\mathcal{S}}, 
\end{equation}
where $\bb_{\mathcal{S}}$ is the set of parameters satisfying \eqref{eq:params_const}. 
\end{thm}

\begin{proof}
In this proof, we will write $\P_{\bb}$ and $\mathbb{A}_{\bb}$ instead of $\P_{\bb_{t}}$ and $\mathbb{A}_{\bb_{t}}$, making the time dependency of the parameters $\bb$ implicit to keep the notation light.

First, we prove that the chain defined by the algorithm in the restricted state space $\mathcal{S}$ is irreducible and aperiodic. The key observation is that starting from any graph it is possible to reach the complete graph $K_n$ in a finite number of moves by adding all the missing edges one by one. These trajectories are possible in the subspace $\mathcal{S}$ since it is closed with respect to edge additions. Furthermore, since every edge addition can be ``reversed'', the corresponding symmetric edge removal is allowed in $\mathcal{S}$. For any pair of connected graphs $G, G' \in \mathcal{S}$, we can build a trajectory in $\mathcal{S}$ between them by first connecting $G$ to $K_n$ using the trajectory described above and then using the reverse trajectory from $K_n$ to $G'$. This property is independent of the choice of parameters $\bb_t$, therefore, the considered Markov chain $\{X_t\}_{t \in \mathbb{N}}$ is thus irreducible. The aperiodicity readily follows from the way we defined $\mathbb{A}_{\bb}$, which assigns a strictly positive probability to stay in one configuration for more than one step.

Our algorithm fits within the framework of Ceperley and Dewing \cite{Ceperley1999}, which we now briefly discuss. The goal of \cite{Ceperley1999} is to determine an explicit expression for the acceptance probability in MCMC in such a way that the algorithm still samples from the desired (stationary) distribution $\pi$ in the setting where the energy difference $\Delta$ between the current state and the next proposed state is perturbed by noise. This also fits our setting, we have $\Delta = \calH_{\bb_{t+1}}(G_{t+1}) - \calH_{\bb_{t}}(G_{t})$ and for each parameter $\beta_j$ we have that $(\beta_j^{t})$ for $t=0,1,2, \dots,T$ is a sequence of random variables. Note that each $\beta_j^{t}$ is adapted to the filtration generated by the initial condition and the acceptance choices up to time $t$. Consequently, the energy difference between graphs $G_t$ and $G_{t+1}$
\begin{align}
    \Delta = \Delta(G_{t},G_{t+1}) := \sum_{j=1}^r \beta_j^{t} (x_j(G_{t+1})-x_j(G_t))
\end{align}
is also random. More specifically, as stated in \cite{Ceperley1999}, a CLT argument suggests that for large enough $T$, for each $j$ the distribution of $\beta_j^t$ is approximately normally distributed with mean $\beta_j$ and variance $\alpha^2$. This is indeed confirmed by the numerical experiments with the update rule \eqref{eq:uprule} conducted in \cite{Borisenko2019}. Assuming that $(\beta_j^t)_{j=1}^r$ are jointly normally distributed (with a nontrivial covariance matrix), we have that
\begin{align}
\Delta \sim \mathcal N (\mu_\Delta, \sigma_\Delta^2),
\end{align}
where $\mu_\Delta := \sum_{j=1}^r\beta_j(x_j(G_{t+1})-x_j(G_t))$. The variance of $\Delta$ is 
\begin{equation} \label{eq:var2}
\begin{split}
\sigma_\Delta^2 := \text{Var}\Big(\sum_{j=1}^r \beta_j^{t} (x_j(G_{t+1})-x_j(G_t))\Big) = & \sum_{j=1}^r \text{Var}\Big(\beta_j^{t} (x_j(G_{t+1})-x_j(G_t))\Big) \\
 & + \sum_{i \neq j} \text{Cov}\Big(\beta_i^{t} (x_i(G_{t+1})-x_i(G_t)),  \beta_j^{t} (x_j(G_{t+1})-x_j(G_t))\Big).
\end{split}
\end{equation}

We aim now to find an upper bound on the variance. Let us consider the case in which the sum of the covariance terms in \eqref{eq:var2} is greater than $0$. Applying the Cauchy-Schwarz inequality, we have the following result:
\begin{equation}\label{eq:cov}
    \small \sum_{i \neq j} \text{Cov}\Big(\beta_i^{t} (x_i(G_{t+1})-x_i(G_t)),  \beta_j^{t} (x_j(G_{t+1})-x_j(G_t))\Big) \leq \sum_{i \neq j} \sqrt{\text{Var}\Big( \beta_i^{t} (x_i(G_{t+1})-x_i(G_t))\Big) \cdot \text{Var}\Big( \beta_j^{t} (x_j(G_{t+1})-x_j(G_t))\Big)}.
\end{equation}
By plugging \eqref{eq:cov} into \eqref{eq:var2} we obtain an upper bound that depends only on the variances of the single differences in the Hamiltonian \eqref{eq:var3}:
\begin{equation} \label{eq:var3}
\begin{split}
\text{Var}\Big(\sum_{j=1}^r \beta_j^{t} (x_j(G_{t+1})-x_j(G_t))\Big) \leq & \sum_{j=1}^r \text{Var}\Big(\beta_j^{t} (x_j(G_{t+1})-x_j(G_t))\Big) \\
 & + \sum_{i \neq j} \sqrt{\text{Var}\Big(\beta_i^{t} (x_i(G_{t+1})-x_i(G_t))\Big) \cdot \text{Var}\Big( \beta_j^{t} (x_j(G_{t+1})-x_j(G_t))\Big)}.
\end{split}
\end{equation}

Notice now that the terms $(x_j(G_{t+1})-x_j(G_t))$ are bounded for each $j$ and for every $t$ since they represent time differences in graph statistics whose maximum can be easily computed given the $j-$th statistic formula. 
We can write
\begin{equation}\label{eq:bound}
    |x_j(G_{t+1})-x_j(G_t)| \leq K_{\max}, \quad \forall j, \forall t .
\end{equation}
With the Hamiltonian defined in \eqref{eq:hamf}, we expect $K_{\max}$ to be the maximum of triangles that could be added/removed by the addition or removal of a single edge, which is of the order $\mathcal{O}(n)$.

Thus, we can substitute their contribution in the variance with the constant $K_{\text{max}}$ obtaining

\begin{equation} \label{eq:var4}
\text{Var}\Big(\sum_{j=1}^r \beta_j^{t} (x_j(G_{t+1})-x_j(G_t))\Big)  \leq \sum_{j=1}^r K_{\text{max}}^2\text{Var}\Big(\beta_j^{t}\Big) + \sum_{i \neq j} K_{\text{max}}^2\sqrt{\text{Var}\Big(\beta_i^{t} \Big) \cdot \text{Var}\Big( \beta_j^{t}\Big)}.
\end{equation}

Notice now that the upper bound for the variance depends solely on the variances of the parameters. With the update rule described in \eqref{eq:uprule} and \cref{alg:final}, the variances of the parameters, as stated in \cite{Borisenko2019}, can be assumed to be approximately equal to the square of the learning rate $\alpha$, leading to the following upper bound 
\begin{equation} \label{eq:varFinal}
\text{Var}\Big(\sum_{j=1}^r \beta_j^{t} (x_j(G_{t+1})-x_j(G_t))\Big) \leq \sum_{j=1}^r K_{\text{max}}\alpha^2 + \sum_{i \neq j} K_{\text{max}}\alpha^2.
\end{equation}
The authors in \cite{Ceperley1999} determine an acceptance probability that satisfies the \textit{average} detailed balance equation of the chain. In our setting, the average detailed balance equation is
\begin{equation}\label{eq:dbee}
    \pi(G_t)  \mathbb{A}_\bb(G_t, G_{t+1}) = \pi(G_{t+1}) \mathbb{A}_\bb(G_{t+1} \rightarrow G_t),
\end{equation}
where $G_t$ and $G_{t+1}$ differ by one edge. When $G_t$ and $G_{t+1}$ differ by more than one edge, the average detailed balance equations are always trivially satisfied. In \eqref{eq:dbee}, $\mathbb{A}_\bb \gstep$ is the \textit{average} acceptance probability of the transition from $G_t$ to $G_{t+1}$, where the expectation is taken with respect to the law of $(\beta_i^{t})_{i=1}^r$. The results in \cite{Ceperley1999} imply that, when $\Delta(G_t,G_{t+1})$ follows a normal distribution with \textit{known} variance $\sigma_\Delta^2$, the acceptance probability 
\begin{equation}\label{eq:solgraphs}
    \mathbb{A}_\bb (G_t, G_{t+1};\sigma_\Delta) =\exp{(-\Delta(G_t,G_{t+1})- \sigma_\Delta^2 / 2)} \wedge 1.
\end{equation}
solves \eqref{eq:dbee}. 
The issue with this is that, in general, $\sigma_\Delta$ in \eqref{eq:solgraphs} cannot be computed explicitly. However, given the upper bound for the variance described in \eqref{eq:varFinal}, if one chooses a small enough learning rate $\alpha$ (possibly at the cost of a larger $t$), the term $\sigma_\Delta^2$ can be neglected, leading to the usual acceptance probability
\begin{equation}\label{eq:solgraphs2}
    \mathbb{A}_{\bb} (G_t, G_{t+1}) = \exp{(-\Delta(G_t,G_{t+1}))} \wedge 1.
\end{equation}
With the acceptance probability defined by \eqref{eq:solgraphs2} and the constraint that every state of the chain trajectory should lie in the subspace $\mathcal{S}$ described before, we can follow the same steps as in \cite{Grey2019}. First, we have the following
\begin{equation}\label{eq:grey1}
    \frac{\P_{\bb}(G_{t+1})}{\P_{\bb}(G_t)} = \frac{\P_{\bb}(\mathbf{G} = G_{t+1} ~|~ \mathbf{G} \in \mathcal{S} )}{\P_{\bb}(\mathbf{G} = G_{t} ~|~ \mathbf{G} \in \mathcal{S} )} = \frac{\P_{\bb}(\mathbf{G} = G_{t+1}, \mathbf{G} \in \mathcal{S} )}{\P_{\bb}(\mathbf{G} = G_{t}, \mathbf{G} \in \mathcal{S} )} \times \frac{\P_{\bb}(\mathbf{G} \in \mathcal{S})}{\P_{\bb}(\mathbf{G} \in \mathcal{S})} = \frac{\P_{\bb}(\mathbf{G} = G_{t+1}, \mathbf{G} \in \mathcal{S} )}{\P_{\bb}(\mathbf{G} = G_{t}, \mathbf{G} \in \mathcal{S} )}.
\end{equation}
By definition of the moves in \cref{alg:final}, we have $\P_{\bb}(\mathbf{G} = G, \mathbf{G} \in \mathcal{S} ) = \P_{\bb}(\mathbf{G} = G)$ because the trajectory is always restricted to lie inside $\mathcal{S}$. Therefore, we rewrite \eqref{eq:grey1} as follows:
\begin{equation}\label{eq:grey2}
    \frac{\P_{\bb}(G_{t+1})}{\P_{\bb}(G_t)} =  \frac{\P_{\bb}(\mathbf{G} = G_{t+1})}{\P_{\bb}(\mathbf{G} = G_{t})},
\end{equation}
which in turn can be plugged in \eqref{eq:solgraphs2} leading to the following acceptance probability:
\begin{equation}\label{eq:solgraphs3}
     \mathbb{A}_\bb (G_t, G_{t+1}) = \min \Big\{ 1, \frac{\P_{\bb}(\mathbf{G} = G_{t+1})}{\P_{\bb}(\mathbf{G} = G_{t})})\Big\},
\end{equation}
which is the same acceptance probability as for the Metropolis-Hastings algorithm without constraint. By construction, this acceptance probability guarantees that $\pi(G) = \P_{\bb_{\mathcal{S}}}(\mathbf{G} = G ~|~ \mathbf{G} \in \mathcal{S})$, which coincides with the desired probability distribution \eqref{eq:ergm_const}. Hence, since the chain is irreducible and aperiodic for $t \rightarrow \infty$, \cref{alg:final} converges to the distribution of interest \eqref{eq:ergm_const} if the learning rate $\alpha$ is close to $0$ with the update rule as in \eqref{eq:uprule} and with the acceptance probability \eqref{eq:solgraphs3}.%
\end{proof}

After having estimated the set of parameters $\bar{\bb}$ using \cref{alg:final}, we can then use \cite[Algorithm 2]{Grey2019} to obtain an ensemble of graphs sampled from the probability distribution $\mathbb{P_{\bb_{\mathcal{S}}}}$ as in \eqref{eq:ergm_const}. This algorithm is a modified Metropolis-Hastings algorithm that generates only connected graphs; for completeness, we report the pseudocode in~\cref{alg:conn}. 

\begin{algorithm}[!ht]
\caption{Connected graph generation~\cite{Grey2019}}\label{alg:conn}
\begin{algorithmic}
\State \textbf{Input:} A connected graph $G_0 = (V,E_0)$, set of parameters $\bb$, maximum number of iterations $T$ 
\For{$t=1, \dots, T$} 
\State Generate a random pair $(i,j)$, $i\neq j$, $i,j \in V$
\If{$(i,j) \in E$}
    \State Remove the edge $(i,j)$ and consider $E_t = E_{t-1}\setminus\{(i,j)\}$, obtaining a new graph $G_t:=(V,E_t)$
    \If{$G_t$ is connected}
        \State With probability $\min \Big\{1,\P_{\bb}(G_t)/\P_{\bb}(G_{t-1}) \Big\}$ accept the graph $G_t$ and save it in graph output list $\mathrm{G}$
    \Else
        \State Reject $G_t$
        \EndIf
\Else
    \State Add the edge $(i,j)$ and consider $E_t = E_{t-1} \cup \{(i,j)\}$, obtaining a new graph $G_t:=(V,E_t)$
    \State With probability $\min \Big\{1,\P_{\bb}(G_t)/\P_{\bb}(G_{t-1}) \Big\}$ accept the graph $G_t$ and save it in graph output list $\mathrm{G}$
\EndIf

\EndFor
\State \textbf{Output:} Graph list $\mathrm{G}$
\end{algorithmic}
\end{algorithm}
\FloatBarrier

\section{Numerical results}
\label{sec:results}

In this section, we give details on the tuning of the proposed ERG model and present some numerical results. More specifically, in \cref{sub:tuning} we detailed our implementation of the proposed variant of the EE algorithm, while in \cref{sub:results} we report several statistics of the synthetic grids obtained with the proposed ERG-based procedure against the reference topology. Lastly, in \cref{sub:ablation} we compare the performance of our proposed ERG model with the simpler ERG model whose Hamiltonian does not have the $1$-triangle and $2$-triangle terms. All the experiments in this section have been performed on a laptop with an Intel(R) Core(TM) i7 CPU and 16 GB of RAM. 

\subsection{ERG parameter tuning}
\label{sub:tuning}
In this section, as an illustration, we briefly show how we tuned the parameters for an ERG model that aims to generate synthetic power grids with topological properties similar to a grid given in input. The grid taken as reference is the \texttt{300 IEEE} network, a medium-sized benchmark grid with 300 nodes and 409 edges (see \cref{tab:allgrids}).

In \cref{tab:hparams} we report the hyperparameters of \cref{alg:final} used for the parameter estimation of the ERG model with Hamiltonian as in \eqref{eq:hamf} and the topology of the \texttt{300 IEEE} network as reference for the values of the observables.

\begin{table}[H]
\centering
\begin{tabular}{c|c|c|c}
\toprule
$T$ & $\alpha$ & $c$ & $\theta$ \\
\midrule
20000000 & 0.001 & 0.001 & 100\\
\bottomrule
\end{tabular}
\caption[Hyperparameters for the \texttt{300 IEEE} network]{Hyperparameters used for the EE algorithm when using the \texttt{300 IEEE} network as reference grid.}
\label{tab:hparams}
\end{table}

Since we have no knowledge of the distribution \eqref{eq:ergm_const} before parameter estimation, in all our experiments we use as a starting point for the estimation algorithm \cref{alg:final} the topology of the reference grid $G_0$. Furthermore, the statement of \cref{thm:main_thm} holds asymptotically as $T \rightarrow \infty$, but in practice we need to resort to an approximation of \eqref{eq:est_lim}. If $T$ is large enough, we can estimate the target parameters as in \cite{Borisenko2019}, using
\begin{equation}
\label{eq:avcomp}
    \bar{\bb} =  \frac{1}{T - t_B} \sum_{t = t_{B+1}}^{T} \bb_{t},
\end{equation}
where $t_B$ is the so-called \textit{burn-in time}, \ie a time after which we say that the chain $\{\bb(t)\}_{t \in \mathbb{N}}$ is roughly independent of the initial condition $\bb_0$. 
%

For simplicity purposes, in all of our experiments, we pick the burn-in time in a heuristic way by setting $t_B = 0.75 \cdot T$, effectively keeping only the last quarter of the trajectories to estimate $\bar{\bb}$. 
\begin{figure}[!h]
    \centering
    \begin{subfigure}[b]{0.9\textwidth}
        \centering
        \includegraphics[width=\textwidth]{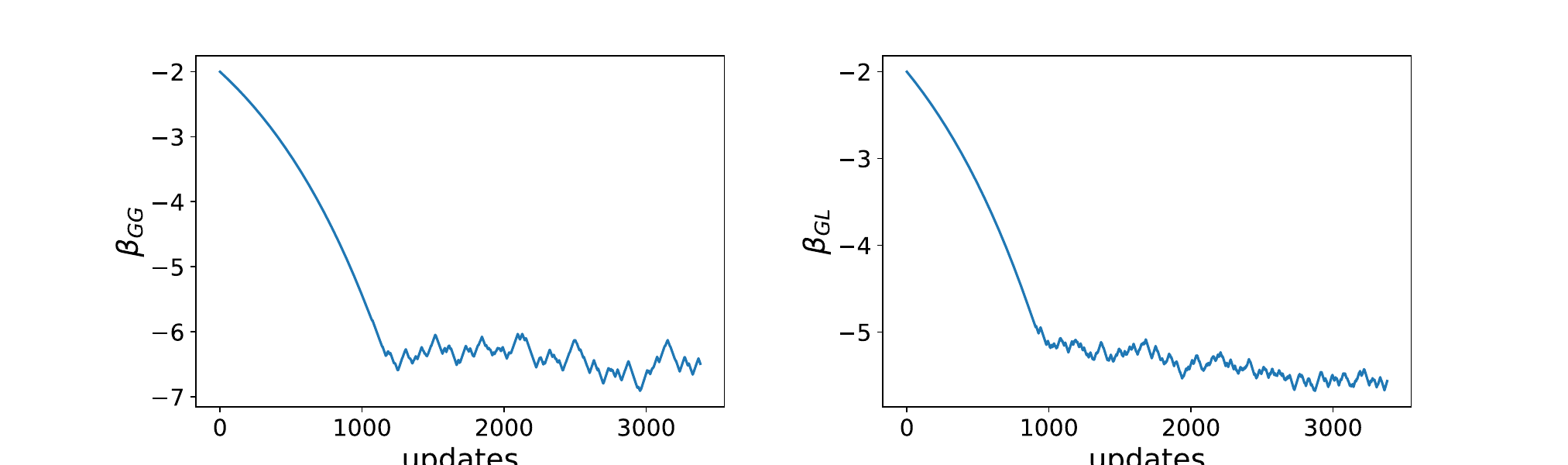}
        \caption{Evolution over time of the parameters $\beta_{PP}$ (left) and $\beta_{PL}$ (right). }\label{fig:traj1}
    \end{subfigure}
    \begin{subfigure}[b]{0.9\textwidth}
        \centering
        \includegraphics[width=\textwidth]{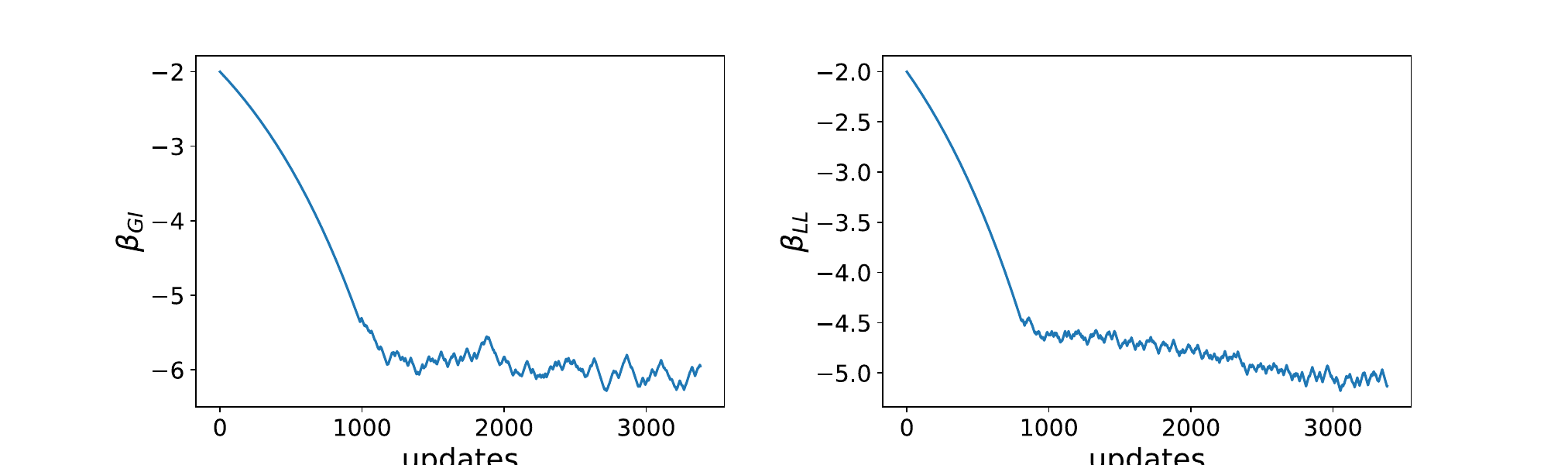}
        \caption{Evolution over time of the parameters $\beta_{PI}$ (left) and $\beta_{LL}$ (right).}\label{fig:traj2}
    \end{subfigure}
    \hfill
    \begin{subfigure}[b]{0.9\textwidth}
        \centering
        \includegraphics[width=\textwidth]{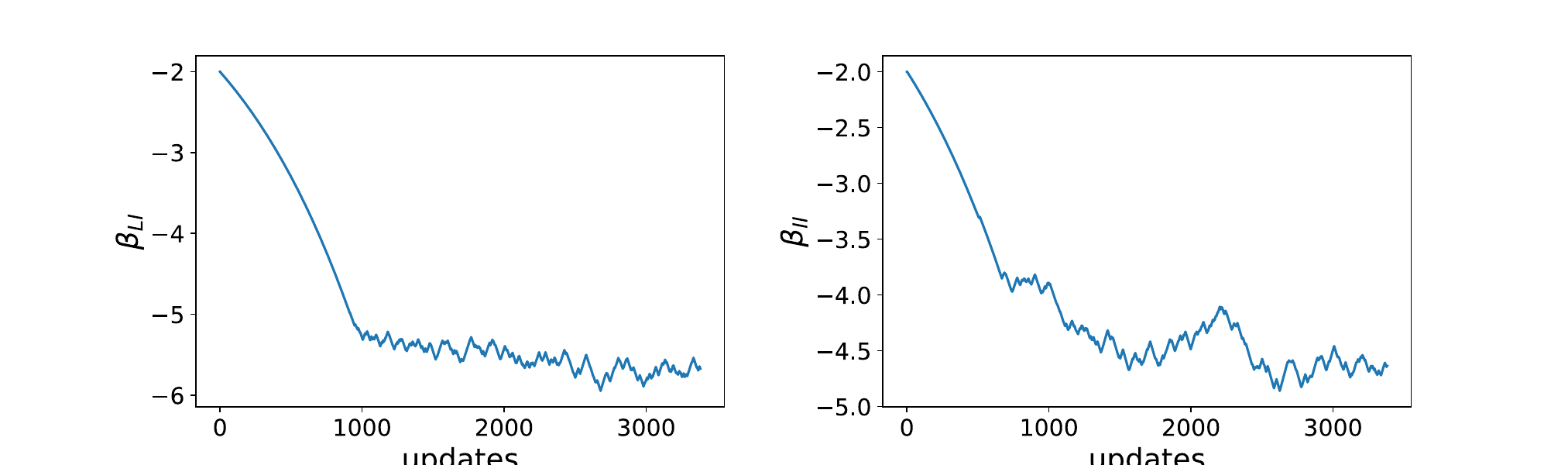}
        \caption{Evolution over time of the parameters $\beta_{LI}$ (left) and $\beta_{II}$ (right).}\label{fig:traj3}
    \end{subfigure}
    \begin{subfigure}[b]{0.9\textwidth}
        \centering
        \includegraphics[width=\textwidth]{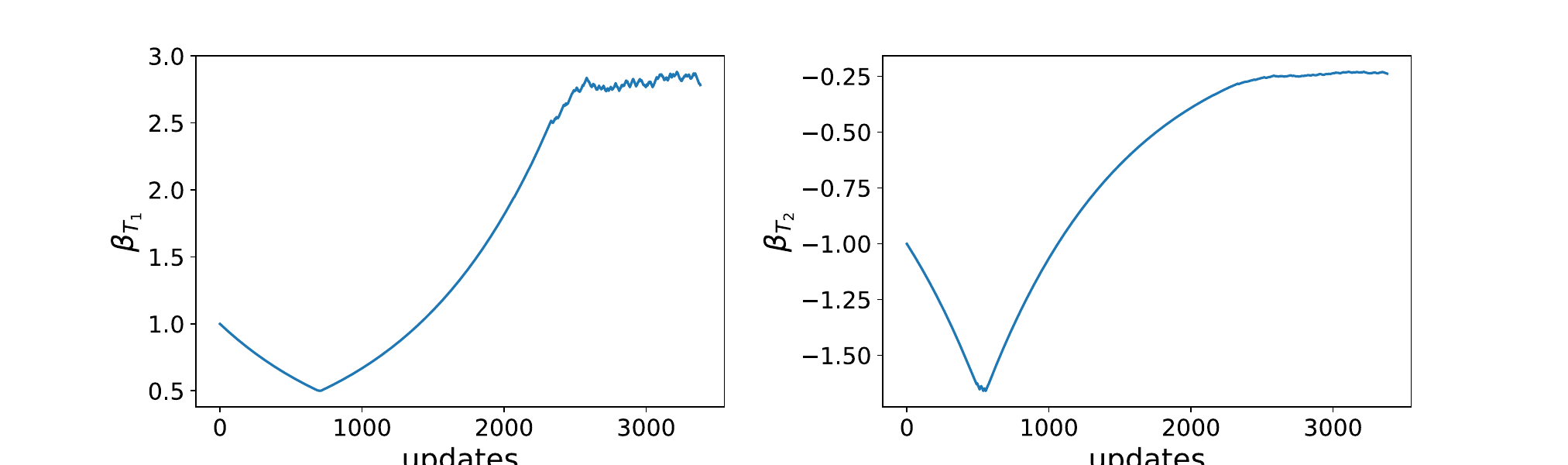}
        \caption{Evolution over time of the parameters $\beta_{1t}$ for the number of triangles (left) and $\beta_{2t}$ for the number of $2-$triangles (right).}\label{fig:traj4}
    \end{subfigure}
    \caption{Trajectories of the $\bb_{k \cdot \theta}, k \in \mathbb{N}$ process describing the ERG model parameters updates over time using the EE algorithm when using the \texttt{300 IEEE} network as the reference grid.}
    \label{fig:trajectories}
\end{figure}

\Cref{fig:trajectories} shows the trajectories of the various components of the vector $\bb_t$ obtained using \cref{alg:final}. For illustration purposes, we display the values only after each nontrivial update (that is, $\bb_{k \cdot \theta}, k \in \mathbb{N}$, cf.~\eqref{eq:uprule}), effectively displaying only the non-piecewise constant parts of the trajectories. For this experiment, the percentage of accepted moves is $\sim 1.5\%$, and the percentage of rejected proposals due to the connectivity constraint is $\sim 0.3\%$. 

Looking at \cref{fig:trajectories}, the burn-in time can be graphically interpreted as the number of update iterations $\bar{t}$ after which each parameter trajectory has passed the ``elbow point'' and starts to oscillate around a constant value.
\FloatBarrier

\subsection{Topological properties of generated synthetic grids}\label{sub:results}
We present here the results of our proposed sampling method on a few benchmark grids of different sizes. For each grid, we sample new graph topologies using the Markov chain described in \cref{alg:conn} with grid-specific parameters $\bar{\bb}$ tuned with \cref{alg:final}. 

To reduce undesirable correlations among the sampled graphs, we applied a standard chain thinning criterion to the list of graphs produced by \cref{alg:conn}. More specifically, we progressively scan this list of samples and retain only the graphs whose adjacency matrix is at a Hamming distance of at least $2n$ from the adjacency matrix of the last retained graph.

In \cref{tab:1.5} we compare the sampled graphs with the initial reference power grids $G_0$, with respect to various topological properties described in~\cref{sec:topo}. We report the average of each quantity over the collection of sampled graphs, as well as the corresponding standard deviation. The whole procedure, complete with parameter estimation and grid ensemble generation, took 25 minutes for the \texttt{118 IEEE} network, 70 minutes for the \texttt{300 IEEE} network, and 250 minutes for the \texttt{1354 PEGASE} network. 

\begin{table}[H]
\centering
\begin{tabular}{p{2.3cm}|c|c|c|c|c|c|c|c}
\toprule
Power grid $G_0$ & Metric type & $m$ & $\langle k_P \rangle $ & $\langle k_L \rangle $ & $\langle k_I \rangle $ & APL & $\lambda_2$ & $C$ \\ \midrule
{\texttt{118 IEEE}} & Actual & 179 & 3.56 & 2.5 & 3.25 & 6.3 & 0.027 & 0.165 \\
\hline
 & Sample avg. & 180 & 3.61 & 2.52 & 3.33 & 4.67 & 0.10 & 0.11 \\
 & Sample std. & (6.89) & (0.36) & (0.20) & (0.70) & (0.42) & (0.04) & (0.03) \\
\midrule
{\texttt{300 IEEE}} & Actual & 409 & 1.95 & 3.00 & 2.15 & 9.93 & 0.009 & 0.085 \\
\hline
 & Sample avg. & 413 & 1.93 & 3.06 & 2.16 & 6.44 & 0.049 & 0.076 \\
 & Sample std. & (10) & (0.13) & (0.17) & (0.24) & (0.39) & (0.019) & (0.015) \\
\midrule
\texttt{1354 PEGASE} & Actual & 1710 & 2.58 & 1.08 & 2.53 & 11.15 & 0.005 & 0.056 \\
\hline
 & Sample avg. & 1723 & 2.53 & 1.16 & 2.51 & 9.56 & 0.022 & 0.051  \\
 & Sample std. & (29) & (0.09) & (0.17) & (0.04) & (0.27) & (0.007) & (0.006) \\
\bottomrule
\end{tabular}
\caption[Comparison of synthetic grid statistics obtained from the ERG model with the corresponding actual values in the reference grid $G_0$.]{Comparison of synthetic grid statistics obtained from the ERG model with the corresponding actual values in the reference grid $G_0$. The reported graph observable are: $m$ the average number of edges, $\langle k_P \rangle $ generators' average degree, $\langle k_L \rangle $ loads' average degree, $\langle k_I \rangle $ interconnections' average degree, APL average shortest path length, $\lambda_2$ algebraic connectivity, and $C$ clustering coefficient. The averages and deviations were computed on an ensemble composed of 1100 samples, obtained after chain thinning.}
\label{tab:1.5}
\end{table}

All generated synthetic topologies are connected, have a realistic bus-type assignment as defined in \cite{Wang2015, Elyas2015}, and exhibit average values close to the ones of the grid of reference for all observables included in the Hamiltonian. Remarkably, even for topological properties that were not explicitly included in our Hamiltonian, \eg the average path length and algebraic connectivity, the generated synthetic grids are within the values regarded as realistic (see \cite{Pagani2013, Wang2017}). Thus, our results are comparable to those reported in works on synthetic grid generation done with reference grids of similar size \cite{Huang2018, Espejo2019, Sadeghian2020}.

\subsection{Ablation study}
\label{sub:ablation}
To show the crucial role that including the number of $1-$ and $2-$triangles as observables in the Hamiltonian plays, we present here a simpler ERG model that only accounts for the edge count and bus-type assignment. Specifically, we consider the following Hamiltonian 
\begin{equation}\label{eq:hams}
    \mathcal{H}_{\bb}(G) \ = \  \beta_{PP}|E_{PP}| + \beta_{PL}|E_{PL}| + \beta_{PI}|E_{PI}| + \beta_{LL}|E_{LL}| +  \beta_{LI}|E_{LI}| + \beta_{II}|E_{II}|,
\end{equation}
which is the same formulation as in \eqref{eq:hamf} but without the $1-$triangles' and $2-$triangles' counts as observables. This simpler Hamiltonian allows us to calculate the partition function $Z_\bb$ explicitly and thus derive in a close form the optimal parameters given the target observables. In fact, the partition function of this simpler model is

\begin{equation}\label{eq:simpz}
    Z_{\bb} = \prod_{i \in I} (1+ e^{\beta_i})^{M(E_i)}, \quad \text{ with } I =\{PP, PL, LI, LL, LI, II\}
\end{equation}
where $M(E_i)$ is a function that computes the maximum possible number of edges of type $i$. In order words, if $E_i=E_{ab}$ is the subset of edges that connect nodes of type $a$ with nodes of type $b$, then
\begin{equation}\label{eq:maxfunc}
    M(E_i)=M(E_{ab}) := \begin{cases}
        |a| \cdot |b| & \mbox{if } a \neq b \\
        |a| \cdot (|a|-1) & \mbox{if } a=b.
    \end{cases}
\end{equation}
Using~\eqref{eq:simpz}, we can then easily derive the following closed-form expression for the free energy:
\begin{equation}\label{eq:simpf}
    F_{\bb} = \log{Z_{\bb}} = \log{\prod_{i \in I} (1+ e^{\beta_i})^{M(E_i)}} = \sum_{i \in I} M(E_i)\log{(1+ e^{\beta_i})}.
\end{equation}
Deriving \eqref{eq:simpf} with respect to each $\beta_i$ we can then calibrate the parameters to match the target value as described in \eqref{eq:params}. The theoretical details of the derivations of \eqref{eq:simpz} and \eqref{eq:simpf} are presented in \cref{app:second}.

However, these close-form parameters do not account for the fact that we are interested only in sampling from the subspace of connected graphs $\mathcal{S}=\calG_{n,\text{conn}}$. Therefore, even for this simpler model, we need to estimate the parameters using the EE algorithm in \cref{alg:final}.

It is insightful to compare the parameters obtained using the closed-form expression of the free energy as in \eqref{eq:simpf} (hence, without taking into account the connectivity constraint) with the ones obtained using the EE algorithm in the subspace of connected graphs. In \cref{tab:parcomparison} we report the values of the different parameters for the two methods when $G_0$ is the topology of the \texttt{300 IEEE} network. Not surprisingly, enforcing connectivity encourages a less sparse solution, and thus the parameter values of the edge counts are lower to account for this. Furthermore, we have noticed that the number of edges in a specific edge subset seems to be correlated with the steepness of the decrease in the value of the associated parameter. Understanding exactly how hard constraints, such as connectivity, affect the parameters of an ERG density is, to the best of our knowledge, still an open problem.

\begin{table}[H]
\centering
\begin{tabular}{c|c|c|c|c}
\toprule
Observable & $G_0$ value &  Free-Energy $\beta$ & EE $\beta$ & Change (in \%) \\ 
\midrule
$e_{PP}$ & 8 & $-5.68$ & $-6.19$ & $-9.0\%$\\
$e_{PL}$ & 110 & $-4.84$ & $-5.22$ & $-7.7\%$ \\
$e_{PI}$ & 9 & $-5.33$ & $-5.83$ & $-9.4\%$ \\
$e_{LL}$ & 240 & $-4.45$ & $-4.60$ & $-3.4\%$ \\
$e_{LI}$ & 35 & $-5.05$ & $-5.38$ & $-6.5\%$ \\
$e_{II}$ & 7 & $-3.89$ & $-4.59$ & $-17.0\%$ \\
\bottomrule
\end{tabular}
\caption[Comparison of parameters for the simpler model obtained with different methods.]{Comparison of parameters for the simpler model in~\eqref{eq:simpz} obtained taking the \texttt{300 IEEE} topology as reference grid $G_0$, using either the closed-form expression of the free energy or the EE algorithm.}
\label{tab:parcomparison}
\end{table}

We show now in \cref{tab:abla} that, in fact, the graphs generated by the simpler model cannot, on average, capture the transitivity of the real power grids. We compare the average results obtained with the model specification as in \eqref{eq:hams} with the one obtained with the specification \eqref{eq:hamf}, once again using the \texttt{300 IEEE} topology as a reference grid $G_0$. The simpler model produces graphs with an average degree per bus-type close to the desired ones, but all the transitivity metrics (\ie clustering coefficient and the related triangle counts) are not consistent with the observed ones.

\begin{table}[H]
\centering
\begin{tabular}{p{4.7cm}|c|c|c|c|c|c|c|c|c}
\toprule
Model & $ m $ & $\langle k_P \rangle $ & $\langle k_L \rangle $ & $\langle k_I \rangle $ & $T_1$ & $T_2$ & APL & $\lambda_2$ & $C$ \\ 
\midrule
Reference grid \texttt{300 IEEE} & 409 & 1.95 & 3.00 & 2.15 & 34 & 14 & 9.93 & 0.009 & 0.085 \\

\midrule
ERG model without $T_1$ and $T_2$ & 407 & 1.88 & 3.00 & 2.13 & 2.7 & 0.08 & 6.00 & 0.067 & 0.005 \\

\midrule
ERG model with $T_1$ and $T_2$ & 413 & 1.93 & 3.06 & 2.16 & 35 & 15 & 6.44 & 0.049 & 0.076 \\

\bottomrule
\end{tabular}
\caption[Comparison of ERG models with the corresponding actual values in the reference grid $G_0$.]{Comparison of the average values of the topological properties of interest between the simpler model arising from \eqref{eq:hams} and the one which includes the triangles arising from \eqref{eq:hamf} w.r.t. the real values observed for the \texttt{300 IEEE} topology.}
\label{tab:abla}
\end{table}

\section{Conclusions and future work}
\label{sec:conclusions}
Exponential Random Graphs are some of the most well-studied random graph models. However, to the best of our knowledge, these have never been used in the context of the generation of synthetic power grids. The advantage of using ERG models in this domain is that desirable local and global topological properties can be introduced simultaneously as soft constraints in the sampling density. The proposed methodology allows to efficiently generate large and diverse samples of synthetic grids starting from a real power grid given as input. In this work, we introduced the ERG mathematical framework, presented an approach to tackle some crucial technical issues, including a general estimation procedure for the parameters which allows for even more flexibility while using the ERG models, and presented some numerical results for the synthetic power grid generation problem. 

Since only a single reference grid is needed as an input to generate an ensemble of weakly correlated topologies that share on average the desired properties, our method can be efficiently used to produce a large collection of realistic grids. These grids in turn, following the rationale of the power grids from a network of networks perspective proposed in \cite{Halappanavar2015}, can be used as \vrg{fragments}, \ie building blocks to be reassembled to generate larger and more heterogeneous grids, as done in \cite{Young2018, Young2018b, Huang2018}. Furthermore, this approach could be used in an iterative fashion, using the reassembled grids as a starting point of reference for \cref{alg:final}.

Possible future research directions include investigating more in depth the dynamics of the chain generated in \cref{alg:final} so that we can identify how well our method scales as the size of the network increases. This would help us identify the mixing time, and thus an appropriate burn-in time $t_B$, in a more theoretically sound way. Finally, since our scheme is rather general, it would be interesting to investigate its application in contexts different from the analysis of power grids. 


\printbibliography

\appendix

\section{Proof of the closed-form expression of the simpler model}
\label{app:second}
\begin{thm}\label{thm:main_thm}
Consider a simple, undirected, unweighted graph $G=(V,E)$, and let $A = \big(A_{ij}\big)$ be the symmetric adjacency matrix associated with $G$. Assume that a partition $E_1,\dots, E_K$ of all possible node pairs $V \times V$ is given, and let $A_1, \dots, A_K$ be the corresponding block decomposition of the adjacency matrix $A$ such that $A_k = \big(A_{ij}\big)_{(i,j) \in E_k}$.
Consider the Exponential Random Graph model on $G \in \mathcal{G}_n$ with Hamiltonian
\begin{equation}\label{eq:ham1}
    \calH_\bb(G) = \sum_{k=1}^{K} \beta_k |E_k(G)|, \quad G \in \mathcal{G}_n,
\end{equation}
where $|E_k(G)|$ indicates the number of nonzero entries in the submatrix $A_k$ for the graph $G$, that is how many edges from the pairs in the subset $E_k$ graph $G$ actually has. Then, the partition function $Z_\bb$ associated with $\calH_\bb$ is of the form
\begin{equation}\label{eq:part1}
    Z_\bb = \prod_{k=1}^{K} (1+ e^{\beta_k})^{M(E_k)},
\end{equation}
where $M(E_k)$ is the total number of entries in the submatrix $A_k$.
\end{thm}
\begin{proof}
In view of the structure of the Hamiltonian \eqref{eq:ham1}, the corresponding partition function is
\begin{equation}\label{eq:zham}
    Z_\bb = \sum_{G \in \mathcal{G}_n} e^{\sum_{k=1}^{K} \beta_k |E_k|}.
\end{equation}
Since we sum over all possible graphs $G \in \mathcal{G}_n$, each block $A_k$ of the adjacency matrix can be considered as an independent matrix and $E_k$ represents the number of edges in the portion of the graph associated with $A_k$. Thus, we can rewrite \eqref{eq:zham} as
\begin{equation*}
    Z_\bb = \sum_{G \in \mathcal{G}_n} \prod_{k=1}^{K} \prod_{A_{ij}\in A_k} e^{\beta_k A_{ij}},
\end{equation*}
Note that, since we consider undirected and unweighted graphs the entries of the adjacency matrix can only take values $0$ or $1$, \ie $A_{ij} \in \{0,1\}$ for all $i, j$.
As described in \cite{Fronczak2012}, since all the considered observables are functions of the entries $A_{ij}$ of the adjacency matrix $A$, we can sum over all possible graphs $G \in \mathcal{G}_n$ by summing over all possible combinations of values for each $A_{ij}$. By doing so, we obtain
\begin{equation*}
    Z_\bb = \prod_{k=1}^K \prod_{A_{ij}\in A_k} (1+e^{\beta_k}) = \prod_{k=1}^K (1+ e^{\beta_k})^{M(E_k)}. \qedhere
\end{equation*}
\end{proof}

\section{List of analyzed grids}\label{app:first}

The transmission grids used in this study are the ones collected and described in \cite{Sogol2019}, which were made available in a \textit{MATPOWER testcase} format \cite{Zimmerman2020} in the library \cite{pg_git} and that we have parsed using the library \cite{ll_git}.
For each available grid, we derive a simple, undirected, and unweighted graph object. The bus types have been inferred as follows: the generators are retrieved directly from the generator list available in the MATPOWER file, the interconnections are the buses with no power generation or demand, while the remaining nodes have all been labeled as loads. In \cref{tab:allgrids} we present some topological properties used during the analysis for all connected power grids.

\begin{table}[H]
\centering
\begin{adjustbox}{width=0.99\textwidth}
\begin{tabular}{lllllllllllllll}
\toprule
                           name &  \#bus & \#branches & $\langle k \rangle$ & $\langle k_P \rangle$ & $\langle k_L \rangle$ & $\langle k_I \rangle$ & \% of $P$ & \% of $L$ & \% of $I$ & \#triangles & \# 2-triangles & $\lambda_2$ &      $C$ &     APL \\
\midrule
        case30\_as &    30 &      41.0 &  2.73 &     2.0 &    2.77 &     4.5 &     20\% &     73\% &      7\% &          6 &             0 &      0.212 &  0.235 &   3.306 \\
      case30\_ieee &    30 &      41.0 &  2.73 &     2.0 &    2.77 &     4.5 &     20\% &     73\% &      7\% &          6 &             0 &      0.212 &  0.235 &   3.306 \\
      case39\_epri &    39 &      46.0 &   2.3 &     1.1 &    2.79 &     0.0 &     25\% &     72\% &      2\% &          1 &             0 &      0.076 &  0.038 &   4.749 \\
      case57\_ieee &    57 &      78.0 &  2.74 &    3.86 &    2.62 &     2.0 &     12\% &     82\% &      5\% &          9 &             2 &      0.088 &  0.122 &   4.954 \\
  case73\_ieee\_rts &    73 &     108.0 &  2.96 &    2.88 &    3.11 &     2.0 &     45\% &     51\% &      4\% &          3 &             0 &       0.04 &  0.025 &   5.983 \\
     case118\_ieee &   118 &     179.0 &  3.03 &    3.56 &     2.5 &    3.25 &     46\% &     47\% &      7\% &         23 &            11 &      0.027 &  0.165 &   6.309 \\
      case179\_goc &   179 &     222.0 &  2.48 &     1.0 &    2.45 &    3.65 &     16\% &     61\% &     22\% &         19 &            13 &      0.007 &  0.089 &  12.382 \\
    case200\_activ &   200 &     245.0 &  2.45 &     1.0 &    2.91 &    3.25 &     24\% &     74\% &      2\% &         13 &             4 &      0.023 &  0.037 &   8.223 \\
    case240\_pserc &   240 &     348.0 &  2.89 &     1.0 &    3.44 &     0.0 &     22\% &     78\% &      0\% &         49 &            33 &      0.017 &  0.114 &   8.824 \\
     case300\_ieee &   300 &     409.0 &  2.73 &    1.96 &    3.06 &    2.15 &     23\% &     68\% &      9\% &         34 &            14 &      0.009 &  0.086 &   9.935 \\
      case500\_goc &   500 &     651.0 &   2.6 &    1.38 &    3.22 &     2.1 &     30\% &     64\% &      6\% &         52 &            24 &      0.007 &  0.061 &    9.75 \\
     case588\_sdet &   588 &     677.0 &   2.3 &    2.54 &    2.18 &    2.87 &     21\% &     72\% &      7\% &          7 &             0 &      0.004 &  0.011 &  13.495 \\
      case793\_goc &   793 &     904.0 &  2.28 &    2.61 &    2.15 &    2.48 &     22\% &     70\% &      8\% &          9 &             0 &      0.003 &   0.01 &  15.331 \\
  case1354\_pegase &  1354 &    1710.0 &  2.53 &    2.58 &    1.08 &    2.53 &     19\% &      1\% &     80\% &         87 &            14 &      0.005 &  0.056 &  11.151 \\
     case2000\_goc &  2000 &    2810.0 &  2.81 &    1.15 &    3.08 &    2.78 &     13\% &     81\% &      6\% &        232 &           129 &      0.001 &  0.063 &  16.363 \\
     case2312\_goc &  2312 &    2830.0 &  2.45 &    2.01 &    2.51 &    2.71 &     18\% &     70\% &     12\% &         52 &            11 &      0.004 &  0.017 &  15.009 \\
     case2383wp\_k &  2383 &    2886.0 &  2.42 &    3.01 &    2.33 &     0.0 &     14\% &     86\% &      0\% &         26 &             3 &      0.003 &  0.009 &  12.759 \\
     case2736sp\_k &  2736 &    3495.0 &  2.55 &    3.44 &    2.45 &     3.0 &     10\% &     90\% &      0\% &         56 &             5 &      0.003 &  0.014 &  13.399 \\
    case2737sop\_k &  2737 &    3497.0 &  2.56 &    3.56 &    2.45 &     3.5 &      9\% &     91\% &      0\% &         57 &             5 &      0.003 &  0.014 &  13.397 \\
     case2742\_goc &  2742 &    4005.0 &  2.92 &    3.67 &    2.91 &    2.07 &      2\% &     98\% &      1\% &        152 &             8 &      0.003 &  0.033 &  15.979 \\
    case2746wop\_k &  2746 &    3505.0 &  2.55 &     3.2 &    2.45 &     2.8 &     14\% &     86\% &      0\% &         59 &             5 &      0.003 &  0.014 &  13.317 \\
     case2746wp\_k &  2746 &    3505.0 &  2.55 &    3.16 &    2.45 &     0.0 &     14\% &     86\% &      0\% &         58 &             5 &      0.003 &  0.014 &  13.302 \\
     case3012wp\_k &  3012 &    3566.0 &  2.37 &    2.96 &    2.29 &    1.22 &     12\% &     88\% &      0\% &         24 &             1 &      0.002 &   0.01 &  14.529 \\
     case3120sp\_k &  3120 &    3684.0 &  2.36 &    2.92 &     2.3 &    1.22 &     11\% &     89\% &      0\% &         25 &             1 &      0.003 &  0.009 &  14.262 \\
     case3970\_goc &  3970 &    5712.0 &  2.88 &    3.34 &    2.86 &    2.54 &      3\% &     97\% &      0\% &        156 &             1 &      0.002 &  0.026 &  17.206 \\
     case4020\_goc &  4020 &    6089.0 &  3.03 &    3.74 &    3.02 &    2.49 &      2\% &     97\% &      1\% &        248 &             3 &      0.003 &  0.038 &  14.679 \\
     case4601\_goc &  4601 &    6305.0 &  2.74 &    3.33 &    2.72 &    3.24 &      3\% &     97\% &      0\% &        117 &             1 &      0.001 &  0.017 &  17.409 \\
     case4619\_goc &  4619 &    7337.0 &  3.18 &    3.78 &    3.19 &     2.5 &      3\% &     93\% &      4\% &        492 &            50 &      0.001 &  0.067 &  18.015 \\
    case4661\_sdet &  4661 &    5751.0 &  2.47 &    2.45 &    2.43 &    2.78 &     20\% &     70\% &     10\% &         92 &            12 &      0.004 &  0.019 &  15.671 \\
     case4837\_goc &  4837 &    6622.0 &  2.74 &    3.22 &    2.72 &    2.54 &      3\% &     96\% &      1\% &        250 &             3 &      0.001 &  0.048 &   23.93 \\
     case4917\_goc &  4917 &    6187.0 &  2.52 &    1.85 &    2.71 &    2.85 &     25\% &     65\% &     10\% &        240 &           147 &      0.001 &  0.035 &  21.466 \\
case6468\_rte\_\_api &  6468 &    8065.0 &  2.49 &    1.83 &     2.6 &    3.27 &     15\% &     84\% &      1\% &        351 &            79 &      0.002 &  0.051 &  14.961 \\
     case6468\_rte &  6468 &    8065.0 &  2.49 &    1.83 &     2.6 &    3.27 &     15\% &     84\% &      1\% &        351 &            79 &      0.002 &  0.051 &  14.961 \\
     case6470\_rte &  6470 &    8066.0 &  2.49 &    1.81 &    2.59 &    4.82 &     15\% &     84\% &      1\% &        352 &            80 &      0.002 &  0.052 &  14.985 \\
     case6495\_rte &  6495 &    8084.0 &  2.49 &    1.78 &    2.56 &    5.95 &     16\% &     83\% &      1\% &        352 &            80 &      0.002 &  0.052 &  14.952 \\
     case6515\_rte &  6515 &    8104.0 &  2.49 &    1.77 &    2.57 &     5.9 &     16\% &     83\% &      1\% &        352 &            80 &      0.002 &  0.051 &  14.952 \\
     case9591\_goc &  9591 &   14042.0 &  2.93 &    3.74 &    2.92 &    2.43 &      2\% &     97\% &      1\% &        557 &            30 &      0.001 &  0.034 &  17.107 \\
    case10000\_goc & 10000 &   12742.0 &  2.55 &    1.33 &    2.69 &    3.57 &     14\% &     81\% &      5\% &        606 &           334 &      0.001 &  0.031 &  23.273 \\
    case10480\_goc & 10480 &   16107.0 &  3.07 &    3.74 &    3.05 &    3.02 &      3\% &     94\% &      2\% &        840 &            81 &      0.001 &  0.053 &  18.606 \\
    case19402\_goc & 19402 &   29751.0 &  3.07 &    3.75 &    2.99 &    3.46 &      2\% &     85\% &     13\% &       1152 &            89 &        0.0 &   0.04 &  20.883 \\
    case24464\_goc & 24464 &   34693.0 &  2.84 &    3.33 &     2.9 &     2.8 &      3\% &     21\% &     76\% &       1346 &           102 &        0.0 &  0.048 &  35.994 \\
    case30000\_goc & 30000 &   35233.0 &  2.35 &    1.54 &    2.36 &    3.75 &      8\% &     89\% &      4\% &        824 &           448 &        0.0 &  0.016 &  34.085 \\
\bottomrule
\end{tabular}
\end{adjustbox}
\caption[Available grids after parsing]{Investigated topological properties of the analyzed grids from \cite{Sogol2019}.}
\label{tab:allgrids}
\end{table}
\end{document}